\documentclass[conference]{IEEEtran}

\makeatletter
\def\ps@headings{
\def\@oddhead{\mbox{}\scriptsize\rightmark \hfil \thepage}
\def\@evenhead{\scriptsize\thepage \hfil \leftmark\mbox{}}
\def\@oddfoot{}
\def\@evenfoot{}}
\makeatother
\usepackage{amsmath}
\usepackage{amssymb}
\usepackage{float}
\usepackage{psfrag}
\usepackage[dvips]{graphicx,color}
\usepackage{cite}
\usepackage{ntheorem}
\newtheorem{theorem}{Theorem}[section]
\newtheorem{lemma}[theorem]{Lemma}
\newtheorem{proposition}[theorem]{Proposition}

\newtheorem{remark}[theorem]{Remark}
\newtheorem{result}[theorem]{Lemma}
\newtheorem{problem}{Problem}
\newtheorem{definition}{Definition}
\newenvironment{proof}[1][Proof]{\begin{trivlist}
\item[\hskip \labelsep {\bfseries #1}]}{\end{trivlist}}
\newenvironment{example}[1][Example]{\begin{trivlist}
\item[\hskip \labelsep {\bfseries #1}]}{\end{trivlist}}
\newcommand{\BigO}[1]{\ensuremath{\operatorname{O}\left(#1\right)}}
\begin{document}
\title{Optimally Approximating the Coverage Lifetime of Wireless Sensor Networks}
\author{
\IEEEauthorblockN{Vivek Kumar Bagaria, Ashwin Pananjady}
\IEEEauthorblockA{Department of Electrical Engineering\\
Indian Institute of Technology Madras\\
Email: \{ee10b047, ee10b025\}@ee.iitm.ac.in}
\and
\IEEEauthorblockN{Rahul Vaze}
\IEEEauthorblockA{School of Technology and Computer Science\\
Tata Institute of Fundamental Research\\
Email: vaze@tcs.tifr.res.in}
}
\maketitle

\begin{abstract}
We consider the problem of maximizing the lifetime of coverage (MLCP) of targets in a wireless sensor network with battery-limited sensors. We first show that the MLCP cannot be approximated within a factor less than $\ln n$ by any polynomial time algorithm, where $n$ is the number of targets. This provides closure to the long-standing open problem of showing optimality of previously known $\ln n$ approximation algorithms.
We also derive a new $\ln n$ approximation to the MLCP by showing a $\ln n$ approximation to the maximum disjoint set cover problem (DSCP), which has many advantages over previous MLCP algorithms, including an easy extension to the $k$-coverage problem. We then present an improvement (in certain cases) to the $\ln n$ algorithm in terms of a newly defined quantity ``expansiveness" of the network.
For the special one-dimensional case, where each sensor can monitor a contiguous region of possibly different lengths, we show that the MLCP solution is equal to the DSCP solution, and can be found in polynomial time. Finally, for the special two-dimensional case, where each sensor can monitor a circular area with a given radius around itself, we combine existing results to derive a $1+\epsilon$ approximation algorithm for solving MLCP for any $\epsilon >0$.
\end{abstract}
\section{Introduction}
\label{intro}
Wireless sensor networks are deployed for a variety of applications - military, data collection, and health-care, to name a few - and most of these entail monitoring or covering a specific geographic area. Therefore, maximizing the lifetime of coverage in a wireless sensor network with battery-limited sensors is a fundamental and classical problem, well studied in literature \cite{cardei2005energy, cardei2006energy, berman2004power, ding2012constant}. 
Typically, a large number of sensors is deployed in a given area, and consequently, many sub-collections of these sensors can cover/monitor all the intended targets. Each such sub-collection of sensors is called a {\it set cover}. 
To maximize the coverage lifetime with the practical constraint of limited battery capacity, we need to find an activity schedule for each sensor (signifying when it must be turned {\it on} or {\it off}) that ensures that all intended targets are covered/monitored for the longest time possible.

Concisely, the maximum coverage lifetime problem (MLCP) is as follows. Given a set of sensors and a set of targets, find an activity schedule for these sensors such that (i) the total time of the schedule is maximized, (ii) all targets are constantly monitored (i.e. at any point of time, at least one of the set covers is active), and (iii) no sensor is used for longer than what its battery allows.

In literature, the MLCP has been approached using two methods. The first method involves solving the maximum disjoint set cover cover problem (DSCP) \cite{bollobas2013cover}. The DSCP finds the maximum number of set covers such that any two set covers are pairwise disjoint. Clearly, sequentially turning on each of the disjoint set covers found by the DSCP provides a feasible solution to MLCP. In \cite{cardei2005improving}, the MLCP was approached using the DSCP. It was proved that the DSCP is NP-complete and shown that the approximation ratio of any polynomial time algorithm to the DSCP has a loose lower bound of 2 (the approximation ratio defines how far from optimum the algorithm's solution is in the worst case). A heuristic algorithm for the DSCP was also provided using an integer program (IP) formulation. In \cite{ahn2011new}, the number of variables and constraints in the IP formulation of \cite{cardei2005improving} was reduced, but the algorithm proposed was still heuristic. Several other heuristic algorithms without provable guarantees have been proposed in literature to solve the MLCP through the DSCP \cite{slijepcevic2001power, lai2007effective, ashouri2012new}. Quite recently, a $\sqrt{n}$-approximation algorithm for the DSCP was proposed in \cite{henna2013approximating} to solve the MLCP, where $n$ is the number of targets to be monitored. However, as shown in \cite{ding2012constant}, the MLCP solution is not always equal to the DSCP solution.

The second method to solve the MLCP uses non-disjoint set covers, i.e. set covers that are not constrained to be disjoint. The optimal solution obtained using this method will exactly match the optimal solution of the MLCP, unlike the DSCP approach. This approach has been used by several papers. Cardei et. al. \cite{cardei2005energy} proved the NP-completeness of the MLCP and formulated it as a linear program (LP). They also provided a few heuristic solutions to the MLCP using non-disjoint set covers. Berman et. al. \cite{berman2004power} also approached the MLCP using non-disjoint set covers, and provided a $1+\ln n$ approximation algorithm to the MLCP by combining the Garg-Koenemann algorithm \cite{garg2007faster} with the $\ln n$ approximation to the minimum weight set cover problem \cite{vazirani2001approximation}. The MLCP has also been approached with additional constraints.
Kasbekar et. al \cite{kasbekar2011lifetime}, considered a variant of the MLCP in which each sensor has information only about its neighbours. They provided a distributed algorithm in which each sensor stays active/inactive depending on the state of its neighbouring sensors, with an $\BigO{\ln n\cdot \ln (nB)}$ approximation ratio, where $B$ is the maximum battery capacity of any sensor.
Zhao et. al. \cite{zhao2008lifetime} solved the MLCP with the additional constraint of connectivity among the sensors and approached the problem using non-disjoint set covers. They posed the modified problem as an LP and obtained a $\ln n$ approximation algorithm.
In \cite{pyun2009power}, the MLCP was modified by taking the energy consumed by data transmission into account. An IP formulation with exponential complexity was proposed, and heuristically solved.

In summary, existing literature on solving the classical MLCP either by the method of disjoint or non-disjoint set covers mainly uses heuristic algorithms \cite{cardei2005energy, slijepcevic2001power, lai2007effective, ahn2011new, cardei2005improving}, and to the best of our knowledge, only \cite{berman2004power} provides provable guarantees on performance for the MLCP without additional constraints.

For the special {\it geometric} case of the MLCP, when each sensor can monitor a circular area around itself with a given radius, a $4+\epsilon$ approximation algorithm was derived in \cite{ding2012constant}, for any $\epsilon > 0$. But in practice, not all coverage problems are geometric, since there could be obstacles or other practical limitations that make the coverage area of any sensor non-convex. In addition, geometric coverage problems need not have circular coverage regions. Therefore, in this paper, we study the general MLCP, as studied by \cite{berman2004power} and \cite{cardei2005energy}.

Even though some approximation algorithms are known for the MLCP, one question that has remained open is how far they are from being optimal. Typically, finding inapproximability results for problems, (i.e. lower bounds on possible approximation ratios) is difficult, and such a result for the MLCP does not exist in prior work. In this paper, for the first time, we show that $\ln n$ is the best possible approximation ratio for the MLCP.

We also derive a new approximation algorithm for the MLCP in this paper by approaching it using the DSCP, and show that it has the best achievable approximation ratio. We also show that it has many advantages over existing algorithms with comparable approximation ratios. 

We state the results of this paper using the following notation and some simple facts. Let $F_{min}$ be the minimum number of sensors that any target is covered by. A target having $F_{min}$ sensors covering it is called a {\it bottleneck} target. With unit battery capacity, $F_{min}$ is an upper bound on the MLCP, since at least one of the sensors containing a bottleneck target must be present in all set covers, and all those sensors can together be used only for time $F_{min}$. Our contributions are as follows.

\begin{enumerate}
\item We show that the MLCP cannot be approximated within a factor less than $\ln n$ in polynomial time unless \linebreak \hbox{$NP \subseteq DTIME(n^{\BigO{\ln\ln n}})$}, which, as with $P=NP$, is widely believed to be false in the theoretical computer science community \cite{gary1979computers}. Thus, we show that the $1+\ln n$ approximation proposed in \cite{berman2004power} is optimal for large $n$. This is the main result of this paper.

\item We propose a polynomial-time approximation algorithm for the DSCP through a suitably defined hypergraph colouring, which returns at least $F_{min}/ \ln n$ disjoint set covers, giving us a $\ln n$ approximation to the DSCP. More importantly, this also gives us a new $\ln n$ approximation algorithm for the MLCP, since with unit battery capacity, operating each of the disjoint set covers for one time unit gives a network lifetime of $F_{min}/\ln n$.
\item We propose a polynomial-time approximation algorithm for the DSCP that returns at least $F_{min}/ \BigO{\ln\Delta_{\tau}}$ disjoint set covers, where $\Delta_{\tau}$ is the expansiveness of any target - the number of other targets with which it is monitored among all sensors - maximized over all targets. This gives us an $\BigO{\ln \Delta_{\tau}}$ approximation algorithm for the MLCP. In certain cases, it is possible that this approximation ratio is better than the worst case ratio of the $\ln n$ obtainable by the hypergraph colouring algorithm of bullet 2.
\item We show that for the one-dimensional case, where each sensor can monitor a contiguous region of possibly different lengths, the MLCP solution is equal to the DSCP solution, and that the MLCP can be formulated as a maximum flow problem on a suitable directed graph, whose solution can be found in polynomial time. This proves a tighter result for the 1-D case than the conjecture in \cite{berman2004power}, which stated that the ratio of optimal solutions of the MLCP and DSCP when sensor coverage areas are convex is upper bounded by $1.5$.

\item For the 2-dimensional geometric case, where each sensor can monitor a circular area around itself with a given radius, we show that a $1+\epsilon$ approximation algorithm exists for any $\epsilon > 0$, using the approach of \cite{berman2004power} together with the $1+\epsilon$ approximation algorithm for finding the minimum weight geometric set cover \cite{mustafa2010improved}.
\end{enumerate}
\section{Preliminaries}
We define a universe of targets $\mathcal{U}=\{1,2,3,\ldots, n\}$, where $n$ is the number of targets. We will hereafter refer to each target as an \emph{element}. Each sensor $i$ can cover a subset of targets $S_i\subseteq \mathcal{U}$, and so the sensors are defined by the multiset $\mathcal{S} = \{S_{1}, S_{2} ,\ldots\}$. Hereon, we use $S_i$ to denote a sensor and call each $S_{i}$ a \emph{subset}. Let each sensor $S_i$ have a battery capacity $b_i$. Since we are interested in monitoring all targets, we define a \emph{set cover} $C \subseteq{\mathcal{S}}$ to be a collection of sensors such that sensors in $C$ cover the universe, i.e., $\bigcup_{S_i\in C}S_i=\mathcal{U}$. The problem is to switch on set covers sequentially so as to prolong the time for which all elements can be monitored (which we call the network lifetime), while ensuring that each sensor is used only for as long as its battery will allow. The formal definition of the MLCP is as follows:
\begin{problem}[MLCP]
Let $\mathcal{C} = \{\,C_1,\ldots,C_m \}$ be the collection of all set covers from $\mathcal{S}$, and $t_i$ be the time for which set cover $C_i$ is switched on. Then the MLCP is to
\begin{eqnarray*}
&&\text{Maximize : } \qquad \sum_{j=1}^m \: t_j \\
&&\text{Subject to } \qquad \sum_{j=1}^{m}C_{ij} t_j \leq b_i, \qquad \forall \; i, \\
&&C_{ij} = 
	\begin{cases}
		\quad 0\text{ if sensor $S_i$ is not in set cover $C_j$}, \\
		\quad 1\text{ if sensor $S_i$ is in set cover $C_j$},
	\end{cases}
\\&&t_j\geq 0 \qquad \forall \; j.
\end{eqnarray*}
\label{def_MLCP}
\end{problem}

We will consider the case where the set $\{b_i\}$ has identical entries $B$, i.e. all battery capacities are equal, which is reasonable. More specifically, if the problem is solved assuming that all battery capacities are equal to 1, it is trivial to see that multiplying all the resulting $t_j$s by $B$ provides the required solution to the MLCP. The case of $b_i\neq b_j$ for some $i\neq j$ will be briefly addressed later, in Remark \ref{LCMbattery}.
\begin{remark}
From here onwards, we consider all sensors to have battery capacity $1$, i.e. $b_i=1$, $\forall$ $i$. \label{battery1}
\end{remark}

\begin{problem}[DSCP]
Given a universe $\mathcal{U}$ and a set of subsets $\mathcal{S}$ as defined above, find as many set covers $C$ as possible such that all set covers are pairwise disjoint (i.e. $C_i\cap C_j=\phi$ $\forall$ $i\neq j$). \label{def_DSCP}
\end{problem}

The DSCP necessitates that any subset can be present in a maximum of one set cover. Note that DSCP=MLCP if $t_j\in \{0,1\}$ $\forall$ $j$ in Problem \ref{def_MLCP} (when $b_i=1$ $\forall$ $i$).

If the number of disjoint set covers is $k$, then using each of the $k$ disjoint set covers for one time unit, clearly, we have an MLCP solution of $k$ with $b_i=1$, $\forall$ $i$. Thus, solving the DSCP provides a feasible solution to the MLCP.
 
However, the optimal solution of the MLCP differs from that of the DSCP as shown in \cite{ding2012constant}, because the optimal solution to the MLCP may not always involve disjoint set covers. For example, let $\mathcal{U}=\{1,2,3\}$ and $\mathcal{S}=\{S_1,S_2,S_3\}$, where $S_{1}=\hspace{-1mm}\{1,2\}$, $S_{2}=\{2,3\}$ and $S_{3}=\{3,1\}$. Clearly, the maximum number of disjoint set covers (and therefore network lifetime) is $1$, while if we operate the sensors as follows: $\{S_{1},S_{2}\}$ for $0.5$ units of time, $\{S_{2},S_{3}\}$ for $0.5$ units of time, and $\{S_{1},S_{3}\}$ for $0.5$ units of time, the lifetime is $1.5$ time units. 

Ironically, however, we show through the proofs of Theorems \ref{lnnharddone} and \ref{MLCPapprox} that in the worst case, the highest network lifetimes obtainable in polynomial time to both the MLCP and DSCP are in fact the same.


Both the DSCP and the MLCP have been shown to be NP-complete, in \cite{cardei2005improving} and \cite{cardei2005energy}, respectively. However, it is possible to \emph{approximate} these problems in polynomial time. A polynomial time approximation algorithm to solve a maximization problem is said to have an approximation ratio $\rho>1$ if it \emph{always} returns a solution greater than $1/\rho$ times the optimal solution. It is said to be a $\rho-$approximation algorithm, or a $\rho$ algorithm.

The main result of this paper is to show that the MLCP cannot have a polynomial time algorithm with an approximation ratio of less than $\ln n$. This solves a long-standing open problem of finding the hardness of the MLCP. In addition, we propose a new algorithm to solve the MLCP by finding an approximation algorithm for the DSCP with optimal approximation ratio. While other algorithms exist which also have a $\ln n$ approximation ratio, they do not use disjoint set covers, which we show to have many advantages.

\subsection{Terminology}
\label{terminology}
\textbf{(i)} \textbf{$n :$} Number of elements in the universe $|\,\mathcal{U}|$.
\textbf{(ii)} \textbf{$|\mathcal{S}| :$} Number of subsets. Note how $\mathcal{S}$ has been defined as a multiset. This is because subset repetitions are possible, since multiple sensors may cover the same targets. $|\mathcal{S}|$ is therefore the \emph{total} number of subsets, not the number of distinct subsets.
\textbf{(iii)} \textbf{$R :$} Maximum size of a subset $S_i$ = $\max\limits_i|\,S_i|$
\textbf{(iv)} \textbf{$F_i :$} The \emph{frequency} $F_i$ of any element $i\in \mathcal{U}$ is defined as the number of subsets $S_j\in|\mathcal{S}|$ that it appears in. $F_i=\#\{S_j:i\in S_j\}$.
\textbf{(v)} \textbf{$F_{min} :$} $\min\limits_i\, F_i$.
\textbf{(vi)} \textbf{$F_{max} :$} $\max\limits_i\, F_i$.
\textbf{(vii)} \textbf{$\Delta_{\tau} :$} \emph{Expansiveness} $\tau_i$ of an element $i\in \mathcal{U}$ is defined as the number of other elements $i$ is present with in all subsets in $\mathcal{S}$. An element $a$ is said to be present with an element $b$ if $\exists$ $S_j\in \mathcal{S}$ such that $a,b\in S_j$. $\Delta_{\tau}=\max\limits_i\,\tau_i$.
\subsection{Some Other Useful Definitions}
We now describe the definitions of a few problems that will occur repeatedly in the paper.
\begin{definition}[Dominating Set]
Given a graph $G = (V,E)$, $V' \subseteq V$ is a dominating set of graph $G$ if for any vertex $v \in V$, either \textbf{(i)} $v \in V'$ or \textbf{(ii)} $v$ is connected to a vertex in $V'$ by an edge. \label{domset}
\end{definition}
\begin{definition}[Domatic Partition]
A domatic partition of graph $G$ partitions the graph into sets \linebreak $\{V_{1}, V_{2} ,\ldots, V_{k}\}$, such that each $V_{i}$ is a dominating set of $G$, and $V_{i}, V_{j}$ are disjoint for $i\neq j$. For the sake of convenience, we call each such dominating set $V_{i}$ obtained as a result of a domatic partition a \emph{domatic set}. \label{dompartition}
\end{definition}
\begin{definition}[Domatic Number]
The domatic number of a graph $G$ is the maximum number of domatic sets obtainable through a domatic partition of $G$. \label{domnumber}
\end{definition} 

In \cite{feige2002approximating}, Feige et. al. explored the domatic number problem on general graphs. They presented an approximation algorithm, which we state here as the following Lemma:
\begin{result}[\cite{feige2002approximating}]
There exists a polynomial time algorithm that returns $\big(\delta + 1\big)/\ln |V|$ domatic sets for any graph $G(V,E)$, where $\delta$ is the minimum degree of $G$. \label{lnndomatic}
\end{result}
Since $\delta+1$ is an upper bound on the domatic number \cite{feige2002approximating}, Lemma \ref{lnndomatic} provides a $\ln |V|$ approximation to the domatic partition problem.

We have laid the basic framework for the rest of the paper. We now prove the hardness of the MLCP.
\section{Hardness of Approximating the MLCP}
\label{MLCPishard}
In this section, we prove one of the main results of this paper - that the MLCP cannot have a polynomial time algorithm which approximates it better than $\ln n$. Before proceeding with the proof itself, we present an overview of approximation and hardness. Note that we will use the term \emph{solution} to mean a feasible solution, and specifically refer to an \emph{optimal solution} when we need to.

\subsection{A note on hardness of approximation}
A problem is said to be \emph{hard to approximate} within a factor $f$, or $f-hard$ to approximate, if the existence of a polynomial time approximation algorithm with approximation ratio $f'<f$ implies certain results in complexity theory that are widely believed to be false (say $P=NP$\;). 

A problem $A$ is said to be reducible to problem $B$ if any instance of problem $A$ can be converted to a particular instance of problem $B$ in polynomial time. The reducibility of an $f-hard$ problem $A$ to a problem $B$ implies that problem $B$ is $f-hard$ to approximate as well.

In the following sections, we use the term \emph{hardness} to mean hardness of approximation. All the hardness results in this section are for problems reduced from the domatic partition problem (Definition \ref{dompartition}). It was shown in \cite{feige2002approximating} that the domatic partition problem is \emph{hard} to approximate within a factor $\ln n$ unless \hbox{$NP \subseteq DTIME\big(n^{\BigO{\ln \ln n}}\big)$}. 

Before getting to the proof of hardness, we define a few terms here for further use in this paper, for the sake of brevity.

\begin{definition}
A problem is said to be $\ln n$ hard if the existence of a $(1-\epsilon)\ln n$ approximation implies that \linebreak \hbox{$NP \subseteq DTIME\big(n^{\BigO{\ln \ln n}}\big)$.}
\end{definition}
\begin{definition}
A problem is said to be \emph{easy}, or \emph{easier} than $\ln n$, if it admits a polynomial-time algorithm with an approximation ratio of $(1-\epsilon)\ln n$ for an $\epsilon>0$.
\end{definition}
\begin{definition}
Problem A is said to be easier than an $f-hard$ problem B if A admits a polynomial time algorithm with an approximation ratio less than $f$.
\end{definition}
Terms such as ``at least as easy" or ``at least as hard" are used as extensions of the above definitions.

The main result of Section \ref{MLCPishard} is summarized in the following Theorem.
\begin{theorem}
The MLCP is $\ln n$ hard to approximate. \label{lnnharddone}
\end{theorem}
The proof of Theorem \ref{lnnharddone} spans Sections \ref{domsection}, \ref{domlifeMLCP} and \ref{endc}. We present an overview of the proof here.

We first introduce a variant of the domatic partition problem called the domatic multi-partition problem. We then extend results from \cite{feige2002approximating} on the domatic multi-partition problem to show that another specially-defined \emph{domatic lifetime} problem is $\ln n$ hard. We then reduce the domatic lifetime problem to the MLCP, and thereby show that the MLCP is $\ln n$ hard. We will use other definitions as we go along to illustrate the reductions, and explain the more intricate details in the following sections.

It will be useful to note the following Lemma, which we will use in the sections that follow. Using the formulation in Problem \ref{def_MLCP}, we define the following:
\begin{definition}\label{utilized}
A set cover $C_k$ is \emph{utilized} in a solution to the MLCP if the time for which it is \textbf{on}, i.e. $t_k>0$.
\end{definition}

\begin{lemma}\label{Scovers}
Any feasible solution to the MLCP in which a polynomial (in $|\mathcal{S}|$ and $n$) number of set covers is utilized and for which the objective function (network lifetime) is $T$ (say), can be used to generate another feasible solution to the MLCP with at most $|\mathcal{S}|$ utilized set covers and objective function $T'\geq T$. 
\end{lemma}
\begin{proof}
Let us consider a feasible solution $X$ to the MLCP in which a polynomial number of set covers (say $\ell_1$) is utilized. Let these set covers be represented by the set $\mathcal{C}_{on}=\{C_1, C_2,\ldots, C_{\ell_1}\}$ (without loss of generality). Let the objective function be $T$. We can now generate the matrix $C$ as in Problem \ref{def_MLCP} in polynomial time using $\mathcal{C}_{on}$ (as $\mathcal{C}$ in Problem \ref{def_MLCP}) and $\mathcal{S}$ and formulate an LP $L_X$ in polynomial time. It is now possible to solve $L_X$ optimally (in polynomial time), to produce an objective function $T'$. Note that $T'\geq T$, by definition. Also note that the LP $L_X$ has $|\mathcal{S}|$ constraints, and so by elementary LP theory, there are a maximum of $|\mathcal{S}|$ variables $t$ that can be non-zero. Let these non-zero variables in the solution of $L_X$ be some $t_{L_1},t_{L_2},\ldots, t_{L_{\ell_2}}$, where $\ell_2\leq |\mathcal{S}|$. This implies that utilizing $C_{on}^*\subseteq C_{on}$ ($C_{on}^*=\{C_{L_1},\ldots C_{L_{\ell_2}}\}$) is sufficient to produce the objective function $T'$.
  \end{proof}

We now present the details of the proof of Theorem \ref{lnnharddone}, starting with the definition of domatic multi-partitioning.

\subsection{Domatic partitioning and multi-partitioning, and Domatic Lifetime} \label{domsection}
In this section, we revisit the domatic partitioning problem (Definition \ref{dompartition}). An extension of the domatic partitioning problem is the domatic multi-partitioning problem, which has been defined in \cite{feige2002approximating} and \cite{fujita2000study} as follows.
\begin{definition}[c-Domatic Multi-partitioning]
Given a graph $G(V,E)$ and a parameter $c$, find as many dominating sets (Definition \ref{domset}) as possible such that each vertex $v\in V$ belongs to at least one but not more that $c$ dominating sets. \label{cdomdef}
\end{definition}

Note that the domatic partitioning problem (Definition \ref{dompartition}) is a special case of c-domatic multi-partitioning in which $c=1$, which requires the dominating sets so obtained to be disjoint. In the domatic multi-partitioning problem, however, disjointness is not a required criterion. 

One can solve the domatic multi-partitioning problem on a graph $G$ by solving the domatic partitioning problem on another graph $G'$. We construct $G'$ from $G$ by the \texttt{DomCopy} procedure, as follows: \textbf{(i)} Make $c$ copies of graph $G(V,E)$. Let the $c$ copies of any vertex $v\in V$ be represented by the set $S_v$. \textbf{(ii)} For every vertex $v\in V$, form a clique between all vertices in $S_v$. \textbf{(iii)} For an edge between vertices $u,v\in G$, form a complete bipartite graph between the vertices in $S_u$ and $S_v$ in $G'$.

The construction of the \texttt{DomCopy} procedure is shown in Fig. \ref{cdomfig} for an original graph $G$ with $V=\{A_1, B_1, C_1\}$ and $c=3$. The perforated lines represent the edges formed after duplication - dots for the cliques (step \textbf{(ii)} of the construction) and dashes for the bipartite graph (step \textbf{(iii)}).
\begin{figure}[h!]
	\begin{center}
	\psfrag{A1}[][][.6]{$A_1$}
	\psfrag{A2}[][][.6]{$A_2$}
	\psfrag{A3}[][][.6]{$A_3$}
	\psfrag{B1}[][][.6]{$B_1$}
	\psfrag{B2}[][][.6]{$B_2$}
	\psfrag{B3}[][][.6]{$B_3$}
	\psfrag{C1}[][][.6]{$C_1$}
	\psfrag{C2}[][][.6]{$C_2$}
	\psfrag{C3}[][][.6]{$C_3$}
	\includegraphics[scale=.24]{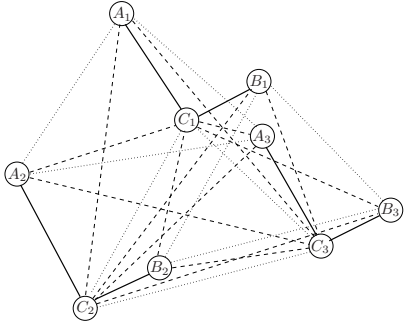} 
	\caption{$G'$ for the 3-domatic problem of graph $A_1, B_1, C_1$, It will be shown in Section \ref{domlifeMLCP} that this corresponds to solving the DSCP on 3 copies each of the subsets $\{A_1, C_1\}, \{B_1, C_1\}$ and $\{A_1, B_1, C_1\}$} \label{cdomfig}
	\end{center}
\end{figure}

\begin{lemma}
Finding the domatic partition of $G'$ is equivalent to finding the c-domatic multi-partition of $G$. \label{domcdom}
\end{lemma}
\begin{proof}
Any dominating set in $G'$ is also a dominating set in $G$ when projected onto it (define projection as considering original copies of those vertices that were copied in the \texttt{DomCopy} procedure). Also, any dominating set in $G$ is a dominating set in $G'$, by the \texttt{DomCopy} construction. Any vertex in $G$ is allowed to be part of $c$ dominating sets in $G'$, by virtue of having $c$ copies in $G'$. We can therefore obtain a c-domatic multi-partition of $G$ by solving the domatic partition problem on $G'$.
     \end{proof}

We now state the following result from \cite{feige2002approximating}.

\begin{result}[\cite{feige2002approximating}] 
The domatic multi-partitioning problem is $\ln n$ hard for any $c$. Here, $n$ represents the number of vertices in the input graph.\label{cdomhard}
\end{result}

For the sake of compactness, we shall now refer to the domatic multi-partitioning problem in which each vertex can be part of $c$ dominating sets as the \emph{c-domatic} problem, and call its optimal solution \emph{c-dom}.

We now prove a result that differs slightly from Lemma \ref{cdomhard}. The distinction will be made clear after the proof.

\begin{lemma}
There must exist a class of c-domatic problems which are $\ln n$ hard, $\forall$ $c$. \label{hardclass}
\end{lemma}
\begin{proof}
This proof requires 2 propositions that follow.
\begin{proposition}\label{DomaticProposition}
Any c-domatic problem for a graph $G$ is at least as easy as finding the domatic number of $G$, $\forall$ $c\geq 2$. 
\end{proposition}
\begin{proof}
By Lemma \ref{lnndomatic}, we have an algorithm which can find $K= (\delta + 1)/ \ln n$ domatic partitions of graph $G$. Recall that $\delta$ was the minimum degree of $G$, and $\delta+1$ was an upper bound on the domatic number of $G$. By the proof of Lemma \ref{domcdom}, a feasible domatic partition of $G'$ involves domatically partitioning each copy of $G$ in $G'$ independently, and so we have an algorithm which finds $c\times{K}$ domatic partitions of $G'$. If we denote the minimum degree of $G'$ by $\delta'$,
\begin{equation}
c\times{K} = \frac{c(\delta + 1)}{\ln n} = \frac{1 + (c\delta+c-1)}{\ln n} = \frac{1 + \delta'}{\ln n},
\end{equation}
since $\delta'=c\delta+c-1$, by the construction of $G'$ from $G$. By definition, the domatic number of $G'$ must be less than $1+\delta'$. We therefore have a $\ln n$ approximation to the c-domatic problem of graph $G$, $\forall$ $c\geq2$. So, any instance of the c-domatic problem is at least as easy as the corresponding domatic partitioning problem (or finding the domatic number).
     \end{proof}
\begin{proposition}
If a c-domatic problem for a graph $G$ becomes easier than $\ln n$ for a particular $c=c_{0}$, it stays easier than $\ln n$ for $G$, $\forall$ $c>c_{0}$ \label{ceasier}
\end{proposition}
\begin{proof}
We will use a logic similar to the one used to prove Proposition \ref{DomaticProposition}.
Let the c-domatic problem become easier than $\ln n$ for a particular $c=c_{0}$. Let us denote the graph constructed for $c=c_{0}$ (by making $c_0$ copies of graph $G$ by \texttt{DomCopy}) by $G_{0}$. We will divide the proof into two cases:\\
\emph{Case 1 :} $c=kc_{0}$, $k\in \mathbb{N}$ : This effectively translates to making $kc_0$ copies of $G$ by the \texttt{DomCopy} procedure, or equivalently, $k$ copies of $G_0$. By Proposition \ref{DomaticProposition}, the problem must be easier than solving $G_0$, which we know is easier than $\ln n$.\\
\emph{Case 2 :} $c\in \{kc_{0}+1, kc_{0}+2 ,\ldots, (k+1)c_{0}-1\}, k\in \mathbb{N}$ : Let us say we are interested in solving the c-domatic problem for some $c=kc_{0}+a$, $a\in\{1,\ldots, c_{0}-1\}$. Let us call the graph generated by making $kc_{0}+a$ copies of $G$ by the \texttt{DomCopy} procedure $G_{req}$. Note that $G_{req}$ is a combination of $k$ copies of $G_0$ and $a$ copies of $G$. If we take $k$ copies of the existing domatic partition of $G_{0}$ (we call this problem 1) and $a$ copies of our domatic partition of graph $G$ (we call this problem 2), we will have a valid domatic partition of $G_{req}$. We know that we had a better than $\ln n$ approximation to problem 1 and a $\ln n$ approximation to problem 2, and so the domatic number of $G_{req}$ has been approximated to a factor better than $\ln n$.
     \end{proof}
\begin{figure}[H]
	\begin{center}
	\psfrag{C=1}[][l][.8]{\footnotesize{$c=1$}}
	\psfrag{C=2}[][l][.8]{\footnotesize{$c=2$}}
	\psfrag{C=3}[][l][.8]{\footnotesize{$c=3$}}
	\psfrag{C=i}[][l][.8]{\footnotesize{$c=i$}}
	\includegraphics[scale=.16]{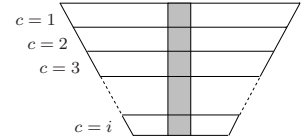}
	\caption{$\ln n$ hard problems for the c-domatic problem. The shaded region represents the set of problems which stay $\ln n$ hard, $\forall$ $c\geq1$} \label{domaticlifetime}
	\end{center}
\end{figure}

Fig. \ref{domaticlifetime} diagrammatically represents Propositions \ref{DomaticProposition} and \ref{ceasier}. The trapezoid for each $c$ represents the class of problems which are $\ln n$ hard for that $c$. Note that such problems might become fewer as $c$ increases (Proposition \ref{DomaticProposition}) and that a $\ln n$ hard problem for a particular $c=c_0$ is also $\ln n$ hard for $c<c_0$ (Proposition \ref{ceasier}). 

We are now in a position to prove Lemma \ref{hardclass}, by contradiction. Let us assume that no graph exists for which the c-domatic problem is $\ln n$ hard for all $c$. This implies that every c-domatic problem $\mathcal{P}$ must be easy for some $c=c_P$. By Proposition \ref{ceasier}, $\mathcal{P}$ must be easy, $\forall$ $c>c_P$. Let $c_{max}=\max_{P}c_P$. If $c_{max}$ is not finite, we have proved Lemma \ref{hardclass}. Otherwise, $\forall$ $c>c_{max}$, there does not exist any graph for which solving the c-domatic problem is $\ln n$ hard. But we know from Lemma \ref{cdomhard} that the c-domatic problem is $\ln n$ hard, $\forall$ $c$. Hence, we have a contradiction.

Therefore, there must exist a class of c-domatic problems which are $\ln n$ hard, $\forall$ $c$. This is represented by the shaded region of Fig. \ref{domaticlifetime}, and Lemma \ref{hardclass} is therefore proved.
  \end{proof}

\begin{remark}
Note the difference between Lemma \ref{cdomhard} and Lemma \ref{hardclass}. Lemma \ref{cdomhard} simply says that there exists some graph $G_c$ for which the c-domatic problem is $\ln n$ hard for any $c$. This graph $G_c$ is c-dependent, and need not be the same for all $c$. Lemma \ref{hardclass} is stronger - it says that there are indeed graphs for which all c-domatic problems (for all $c$) are $\ln n$ hard. 
\end{remark}

We are now ready to define the \emph{domatic lifetime problem}.

\begin{definition}[Domatic Lifetime]
Given a graph $G(V,E)$, maximize the quantity $\emph{(c-dom)}/c$ over all $c$, where $\emph{c-dom}$ is defined as the optimal solution to the c-domatic problem. \label{DomLife}
\end{definition}
Note here that solving the domatic lifetime problem \emph{need not necessarily} involve iterating through all values of $c$ to solve the c-domatic problem. Given Definition \ref{DomLife}, we first observe the following:
\begin{lemma}
It is hard to approximate the domatic lifetime problem within $\ln n$. \label{domaticlemma}
\end{lemma}
\begin{proof}
Let us take a graph $G$ from the class defined in Lemma \ref{hardclass}. Let us say we had access to an oracle, which tells us the value of $c$ for which $\emph{c-dom}/c$ was maximized. Even with this information, note that we cannot do better than a $\ln n$ approximation to the domatic lifetime problem, since that would allow us to approximate $\emph{c-dom}$ within a factor of $\ln n$. Even with access to the oracle, the domatic lifetime problem is $\ln n$ hard. So the domatic lifetime problem for all such graphs $G$ in the hard class must be $\ln n$ hard. 
  \end{proof}

A reduction from the domatic lifetime problem to the MLCP would therefore prove Theorem \ref{lnnharddone}, by Lemma \ref{domaticlemma}. We attempt to do that in the next 2 sections.

\subsection{From Domatic Lifetime to the MLCP} \label{domlifeMLCP}
We first define another related problem - the MLCP-T - to bridge the gap between the domatic lifetime problem and the MLCP.
\begin{definition}
An MLCP-T problem is an MLCP in which each set cover utilized is \textbf{on} for exactly the same time $T$. Set covers in the solution may be repeated. \label{Tdef}
\end{definition}
\begin{example} \label{Texample}
Consider $\mathcal{U}=\{1,2,3\}$ and the set of subsets $\mathcal{S}=\{\{1,2\}, \{2,3\}, \{3,1\}, \{1,2,3\}\}$. By inspection, it ought to be clear that using the set covers $C^1_1=\{\{1,2\}, \{2,3\}\}$, $C^1_2=\{\{2,3\}, \{3,1\}\}$, $C^1_3=\{\{1,2\}, \{3,1\}\}$ and $C^1_4=\{\{1,2,3\}\}$ with $t^1_1=t^1_2=t^1_3=0.5$ and $t^1_4=1$ is a solution to the MLCP. The solution to the MLCP-T will be $C^2_1=\{\{1,2\}, \{2,3\}\}$, $C^2_2=\{\{2,3\}, \{3,1\}\}$, $C^2_3=\{\{1,2\}, \{3,1\}\}$, $C^2_4=\{\{1,2,3\}\}$, and $C^2_5=\{\{1,2,3\}\}$ with $t^2_1=t^2_2=t^2_3=t^2_4=t^2_5=T=0.5$. Notice how we split $C^1_4$ into $C^2_4$ and $C^2_5$ in order to ensure that all set covers were operated for the same amount of time $T$.
\end{example}
The connection between the solutions of the MLCP and MLCP-T is not a coincidence. We will explore this connection more rigourously in Section \ref{endc}.
We are now in a position to define a suitable reduction from the domatic lifetime problem. Our objective in this section is summarized as follows.
\begin{lemma}
The domatic lifetime problem can be reduced to the MLCP-T. \label{reducibility} 
\end{lemma}
The proof of Lemma \ref{reducibility} will involve reductions spanning Propositions \ref{domDSCP}, \ref{multiDSCP} and \ref{MLCPDSCP}. A brief roadmap of the reduction is schematically shown in Figure \ref{relationship}. It will be useful to refer to Fig. \ref{relationship} while working through the reduction; some problems mentioned in the figure will be defined as we go along.

\begin{figure}[H]
    \begin{center}
    \psfrag{DCSP}[][][.7]{DSCP}
    \psfrag{cDSCP}[][][.7]{c-DSCP}
    \psfrag{Domatic}[][][.7]{Domatic}
    \psfrag{cDomatic}[][][.7]{c-Domatic}
    \psfrag{MLCP}[][][.7]{MLCP}
    \psfrag{MLCPT}[][][.7]{MLCP-T}
    \psfrag{DomaticLifetime}[][][.7]{Domatic Lifetime}
    \psfrag{some}[][][.6]{for some $c$}
    \psfrag{co}[][r][.6]{    *solutions coincide}
    \hspace*{-.05in}
    \includegraphics[scale=.65]{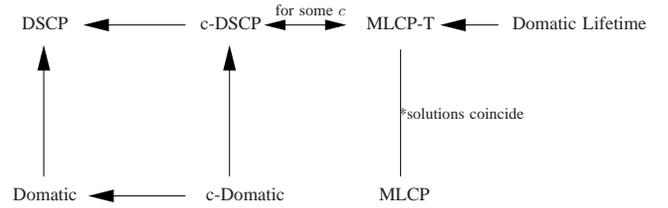}
    \caption{Relationship between the problems in Section \ref{MLCPishard}. $A\rightarrow B$ signifies that problem $A$ is reducible to problem $B$, and $A\longleftrightarrow B$ signifies that problems $A$ and $B$ are equivalent. The relationship between the MLCP and MLCP-T will be handled in Section \ref{endc}} \label{relationship}
    \end{center}
\end{figure}
Let us start by relating problems on graphs (the domatic problems) to our inputs of $\mathcal{U}$ and $\mathcal{S}$.
\begin{proposition}
Every domatic partitioning problem (refer Definition \ref{dompartition}) can be reduced to a DSCP. \label{domDSCP}
\end{proposition}
\begin{proof} 
We prove this by what we call the \texttt{GraphToSet} procedure. Consider the graph $G(V,E)$. Let $V=\{v_{1}, \ldots, v_{n}\}$. Let the \emph{neighbourhood} of any vertex $v_{i}\in V$, represented by $N(v_{i})$, be the set of vertices connected to $v_{i}$ by an edge. Let $N^{+}(v_{i})$ be the union of $N(v_{i})$ and $v_{i}$, i.e., the \emph{inclusive neighbourhood} of $v_{i}$.

We wish to formulate an equivalent DSCP with a universe $\mathcal{U}$ and set of subsets $\mathcal{S}$. Let $\mathcal{U}=V$ and \linebreak \hbox{$\mathcal{S}=\{N^{+}(v_{1}), N^{+}(v_{2}),\ldots, N^{+}(v_{n})\}$}. We can see that the DSCP on $\mathcal{U}$ and $\mathcal{S}$ is equivalent to the domatic partition of the graph $G$. An example is shown in Fig. \ref{cdomfig}.
     \end{proof}
Similarly, the c-domatic problem is also expected to be reducible to some problem involving a suitably defined $\mathcal{U}$ and $\mathcal{S}$. We now show that that problem is the \emph{c-DSCP}.
\begin{definition}[c-DSCP]
The c-DSCP given a universe $\mathcal{U}$ and a collection of subsets $\mathcal{S}$ is equivalent to the DSCP on $\mathcal{U}$ and set of subsets $\mathcal{S'}$, where $\mathcal{S'}$ contains each of the subsets $S_i\in \mathcal{S}$ repeated $c$ times.
\end{definition}
\begin{proposition}
Every c-domatic problem can be reduced to the c-DSCP. \label{multiDSCP}
\end{proposition}
\begin{proof}
The proof is similar to that of Proposition \ref{domDSCP}. We again use graph $G$ and $G'$ constructed as the c-domatic version of $G$ (refer to the \texttt{DomCopy} procedure in Section \ref{domsection}). Let us denote the set of vertices in $G$ by $V$ and the vertices in $G'$ by $V'$. We consider the inclusive neighbourhoods $\mathcal{S'}=\{N^{+}(v_{1}), N^{+}(v_{2}),\ldots\}$ of the vertices $v_i\in V'$. By the construction of $G'$, $\mathcal{S'}$ is effectively $c$ copies of each of $N^+(v_i)$, $\forall$ $v_i\in V$. Solving the DSCP on $\mathcal{S'}$ gives us the domatic partition of $G'$ by Lemma \ref{domDSCP}. We know from Lemma \ref{domcdom} that the domatic partition of $G'$ is equivalent to the c-domatic multi-partition of $G$. Therefore, the c-DSCP on $\mathcal{S}=\{N^{+}(v_{1}), N^{+}(v_{2}),\ldots, N^{+}(v_{n})\}$, $v_i\in V$ is equivalent to the c-domatic problem on $G$. Fig. \ref{cdomfig} illustrates this.
     \end{proof}

Now we look at how the MLCP-T problem can be solved on a given universe $\mathcal{U}$ and set of subsets $\mathcal{S}$.
\begin{proposition}
Solving the MLCP-T is equivalent to solving the c-DSCP for some $c$. \label{MLCPDSCP}
\end{proposition}
\begin{proof}
%
Any feasible MLCP-T solution on $\mathcal{U}$ and $\mathcal{S}$ will consist of some set covers $\mathcal{C}_{T_0}$, each operated for time $T_0$. The total network lifetime is therefore $|\mathcal{C}_{T_0}|T_0$. Note that any sensor in $\mathcal{S}$ can appear a maximum of $c_0=1/T_0$ times in the set covers in $\mathcal{C}_{T_0}$, because of the unit battery constraint. If we now use the set of subsets $\mathcal{S}'=\mathcal{S}$ repeated $c_0$ times, we see that $\mathcal{C}_{T_0}$ would be a feasible solution to the DSCP solved on $\mathcal{U}$ and $\mathcal{S}'$, or in other words, the $c_0$-DSCP on $\mathcal{U}$ and $\mathcal{S}$.

Let the optimal solution of the MLCP-T consist of the collection of set covers $\mathcal{C}_{T_1}$, where each set cover is used for time $T_1$. This can similarly be thought of as solving the c-DSCP for some $c_1=1/T_1$. Note that the optimal network lifetime of the MLCP-T is therefore $|\mathcal{C}_{T_1}|T_1$. If the optimal solution of the \hbox{$c_1$-DSCP} was some collection of set covers $\mathcal{C}_{c_1}$ such that $|\mathcal{C}_{c_1}|>|\mathcal{C}_{T_1}|$, then we could use those $|\mathcal{C}_{c_1}|$ set covers each for time $T_1$ (since there are still moat most $1/T_1$ set covers containing any sensor) and obtain a network lifetime higher than $|\mathcal{C}_{T_1}|T_1$ for the MLCP-T. This contradicts the fact that the optimal network lifetime is $|\mathcal{C}_{T_1}|T_1$.
     \end{proof}

All that is left to prove Lemma \ref{reducibility} is to tie up Propositions \ref{multiDSCP} and \ref{MLCPDSCP} to the domatic lifetime problem. 

We know from the proof of Lemma \ref{MLCPDSCP} that solving the MLCP-T is equivalent to solving the c-DSCP for some $c=c_1$. The optimal solution to the MLCP-T thus obtained will be equal to $|\mathcal{C}_{c_1}|T_1= |\mathcal{C}_{c_1}|/c_1$, where $\mathcal{C}_{c_1}$ is the collection of set covers obtained as the solution of the $c_1$-DSCP, and $T_1=1/c_1$ is the time for which each set cover is utilized in the optimal MLCP-T solution.

As proved in Lemma \ref{multiDSCP}, the \emph{c-domatic} problem is reducible to c-DSCP. When we allow $c$ to vary such that the $\emph{(c-dom)}/c$ is maximized (the domatic lifetime problem), we are effectively finding the $\max_c |\mathcal{C}_{c}|/c$ for the c-DSCP reduced from the c-domatic problem. This is the optimal solution to the MLCP-T, by the proof of Proposition \ref{MLCPDSCP}.

So the domatic lifetime problem has been reduced to the MLCP-T, and this proves Lemma \ref{reducibility}.

Lemmas \ref{reducibility} and \ref{domaticlemma} together lead to the following Lemma.
\begin{lemma}
The MLCP-T problem is $\ln n$ hard. \label{Thard}
\end{lemma}

We are now left with the task of tying up the results of Section \ref{domlifeMLCP} with the MLCP. We do that in the next section.
\subsection{Connecting the MLCP-T to the MLCP} \label{endc}
The connection between the MLCP and MLCP-T was briefly touched upon in Example \ref{Texample}. We will explore it more closely now, using Definition \ref{coincide}.
\begin{definition}[Coincidence of Solutions]\label{coincide}
The solutions of two problems are said to coincide if a feasible solution to one problem can be converted into a feasible solution to the other problem without changing the objective function (not necessarily in polynomial time).
\end{definition}
\begin{lemma}
The solutions of the MLCP and MLCP-T problems coincide. \label{Teq}
\end{lemma}
\begin{proof}
A feasible solution to the MLCP-T problem is also a feasible solution to the MLCP, by definition. To show the converse, let us take a solution of the MLCP that we wish to convert into a solution to the MLCP-T. Without loss of generality, let all $t_i$s be rational, since the optimal solution of the MLCP will be a corner point solution of the LP in Problem \ref{def_MLCP} and thus have rational $t_i$s. In order to convert this to a feasible MLCP-T problem, take the least common denominator (LCM) of all $t_i>0$, and call this $c$. Let $T=1/c$. Note that all $t_i$s are now multiples of $T$. Now create a solution to the MLCP in which $t_i/T$ copies of set cover $C_i$ appear. Each set cover in the new solution is therefore utilized for time $T$, and this is a feasible MLCP-T solution with the same objective function.
     \end{proof}
While Lemma \ref{Teq} may seem trivial to the reader, we provide a warning here concerning conclusions that must be drawn from it.
\begin{remark}
The coincidence of solutions of the MLCP and MLCP-T should not be interpreted as equivalence (two sided reduction implies equivalence). Recall that a reduction is valid only if it can be carried out in polynomial time. Note here, however, that the parameter $T$ (or equivalently, $c$) poses problems. Depending on the problem instance, the parameter $c=1/T$ may be exponentially large in $\mathcal{S}$, in which case making copies of the set covers is not a polynomial time operation. \label{crem}
\end{remark}

Remark \ref{crem} basically tells us that Lemma \ref{Thard} is insufficient to prove that the MLCP is $\ln n$ hard, since the reduction from the MLCP-T to the MLCP cannot be accomplished in polynomial time.

It is therefore necessary to carefully handle the parameter $c$ in the reduction from MLCP-T to MLCP. Note here that $c$ depends on the problem instance, but does not directly depend on the input parameters $n$ or $|\mathcal{S}|$. In this section, we illustrate a method by which the difficulty posed by $c$ can be circumvented. This technique may be useful in analyzing the hardness of other suitably formulated LPs for NP-hard problems.

Our method will provide a proof of Theorem \ref{lnnharddone} by contradiction, using the Lemmas and Propositions in Sections \ref{domsection} and \ref{domlifeMLCP}. Let us assume that Theorem \ref{lnnharddone} is false, i.e., let there exist some algorithm \texttt{DEFY}, which runs in polynomial time and produces a solution $X>OPT/\ln n$ for all input instances of the MLCP ($OPT$ here refers to the optimal solution of the MLCP). Since \texttt{DEFY} runs in polynomial time, it must output a polynomial number of utilized set covers (refer Definition \ref{utilized}). By Lemma \ref{Scovers}, we know that we can consider the number of utilized set covers to be at most $|\mathcal{S}|$. Consider these $|\mathcal{S}|$ set covers to be the solution of \texttt{DEFY}, represented by the set $\mathcal{C}_{on}$, where the time for which set cover $C_i\in \mathcal{C}_{on}$ is on is $t_i$. Recall that $|\mathcal{C}_{on}|\leq|\mathcal{S}|$.

To handle the parameter $c$, we now propose a \texttt{PrecisionControl} procedure. We first make a few observations: 

1. All $t_i$s corresponding to the set covers $C_i\in\mathcal{C}_{on}$ are rational numbers. This is because the $t_i$s were obtained as a corner point solution of an LP with $|\mathcal{S}|$ constraints, each of the form $A\bf{t}\leq 1$, where $A$ is some $0-1$ row matrix. All $t_i$s will therefore be of the form $n/d$, where $d\in\{1,2,\ldots, |\mathcal{S}|\}$ and $n\leq d$. 

2. The least common denominator of all $t_i$s is in the worst case $LCM(1,2,\ldots, |\mathcal{S}|)$, which grows exponentially in $|\mathcal{S}|$. The parameter $c$ in the reduction from MLCP-T to MLCP can therefore be exponential in the input size. 

To work around this, let us fix a minimum precision (maximum resolution) $\epsilon/|\mathcal{S}|^2$ and truncate all $t_i$s to that precision. Note that we are effectively enforcing that the LCM $c$ in the proof of Lemma \ref{Teq} to be $\epsilon/|\mathcal{S}|^2$, which is polynomial in $|\mathcal{S}|$. The most we lose in terms of the solution is $\epsilon/|\mathcal{S}|$, since there are $|\mathcal{S}|$ $t_i$s and each is reduced by at most $\epsilon/|\mathcal{S}|^2$. So the solution outputted by \texttt{DEFY} is now $X'=X-\epsilon/|\mathcal{S}|>OPT/\ln n$. The approximation ratio is still greater than $\ln n$ even after the \texttt{PrecisionControl} procedure.\footnote{Note than $\epsilon/|\mathcal{S}|^2$ could have been replaced by $\epsilon/|\mathcal{S}|^\alpha$ for any $\alpha>1$ to ensure that the approximation ratio of \texttt{DEFY} stayed greater than $\ln n$.} The change now is that a conversion to the MLCP-T solution will involve $T=\frac{1}{c}=\frac{1}{\epsilon/|\mathcal{S}|^2}$. 

We can use the \texttt{DEFY} algorithm to construct a better-than-$\ln n$ approximation to the MLCP-T problem in polynomial time. To show this contradiction to Lemma \ref{Thard}, let us start by taking a graph $G$ in the grey region of Fig. \ref{domaticlifetime}, for which, by Lemma \ref{domaticlemma}, the domatic lifetime problem is $\ln n$ hard. By the \texttt{GraphToSet} procedure, construct a universe $\mathcal{U}_G$ and a set of subsets $\mathcal{S}_G$. By Lemma \ref{Thard}, the MLCP-T problem on $\mathcal{U}_G$ and $\mathcal{S}_G$ must be $\ln n$ hard. Let the optimal solution of the MLCP-T (and MLCP, by Lemma \ref{Teq}) be represented by $OPT(G)$. Use the \texttt{DEFY} algorithm with the \texttt{PrecisionControl} procedure to solve the MLCP on $\mathcal{U}_G$ and $\mathcal{S}_G$ and let the solution returned be $X(G)>OPT(G)/\ln n$. Note that this solution occurs at $T=\epsilon/|\mathcal{S}_G|^2$ in the context of the MLCP-T. A solution to the MLCP-T can therefore be outputted by making a maximum of $|\mathcal{S}_G|^2/\epsilon$ copies of any set cover in the MLCP solution. A better than $\ln n$ approximation to the MLCP-T can therefore be generated in polynomial time. To sum it up, with control over $T$, we can generate a better than $\ln n$ approximation to the MLCP-T in polynomial time using the better than $\ln n$ approximation to the MLCP. This is a contradiction to Lemma \ref{Thard}.

Theorem \ref{lnnharddone}, with which we started Section \ref{MLCPishard}, is therefore proved.

We have therefore proved that the $1+\ln n$ approximation algorithm derived in \cite{berman2004power} cannot be improved upon in polynomial time unless \hbox{$NP \subseteq DTIME(n^{\BigO{\ln\ln n}})$}. In the next section, we present a few novel algorithms which use the DSCP to approximate the MLCP optimally. We then present some advantages of our algorithms over existing optimal algorithms.

\section{Algorithms for the MLCP} \label{MLCPalgos}
In this section, we first provide algorithms for approximating the DSCP. We then show that these algorithms are also equivalent approximations to the MLCP. In prior work, this approach has only led to heuristic algorithms \cite{cardei2005improving, ahn2011new, slijepcevic2001power, lai2007effective}, but we provide the optimal $\ln n$ approximation using this method.

We recall to the reader Remark \ref{battery1}, by which we assumed that battery capacities $b_i$ of all sensors $S_i$ are equal to unity (or more generally, equal), and so justified approaching the MLCP through the DSCP. The natural question that arises is what if $b_i\neq b_j$ for some $i\neq j$. We will briefly address that in the following remark.
\begin{remark}
If $b_i\neq b_j$ for some $i\neq j$, take the greatest common divisor (GCD) of all $\{b_i\}$ and call this $B_{com}$. Create a new set of sensors $\mathcal{S}'$ containing $b_i/B_{com}$ copies of each sensor $S_i\in\mathcal{S}$. Each sensor in $\mathcal{S}'$ can now be thought of as having equal battery capacity $B_{com}$, a case addressed by Remark \ref{battery1}. Note that this will only be a polynomial time operation if $\max_i b_i/B_{com}$ is of polynomial size in $n$ and $|\mathcal{S}|$. For the general case of $b_i\neq b_j$ for some $i\neq j$, the results of Section \ref{MLCPalgos} apply only if $\max_i b_i/B_{com}=poly(n,|\mathcal{S}|)$. \label{LCMbattery}
\end{remark}

With that minor question addressed, we now return to the DSCP assuming Remark \ref{battery1}. There are a few observations to be made about the DSCP (refer Problem \ref{def_DSCP} for formal definition) with respect to the terminology we defined in Section \ref{terminology}. The optimal solution to the DSCP cannot be more than $F_{min}$, since the element with frequency $F_{min}$ can only be present in a maximum of $F_{min}$ disjoint set covers. We therefore have a trivial upper bound of $F_{min}$ on the optimal solution to the DSCP.

The DSCP is NP-complete \cite{cardei2005improving}, but we show that when $F_{max}=F_{min}=2$, a polynomial time solution to the DSCP exists. This algorithm is presented in Appendix \ref{fmaxfminappendix}.

We now present algorithms for the DSCP when $\mathcal{U}$ and $\mathcal{S}$ can be arbitrary. The first algorithm uses an elegant result from the domatic number problem (Definition \ref{domnumber}).
\subsection{$\ln \big(|\mathcal{S}|+n\big)$ approximation to the DSCP}
We use the domatic partitioning (domatic number) problem (refer Defs. \ref{dompartition} and \ref{domnumber}) in this algorithm. Recall that the domatic partitioning problem involved partitioning a graph into dominating sets, each of which we called a \emph{domatic set}.

We will now show that the DSCP is closely related, but not equivalent, to the domatic partitioning problem. Our purpose in this section is to try and solve the DSCP as a domatic partitioning problem on a suitably defined graph $\mathcal{R}$ as follows.

\begin{definition}[$\mathcal{R}(V,E)$] 
Each subset $S_i$ is represented by a vertex $v_i\in V$. Two vertices $v_i$ and $v_j$ are connected by an edge $e\in E$ if $S_i$ and $S_j$ share at least one element, i.e. $S_i\bigcap S_j\neq \phi$. \label{simplegraph}
\end{definition}

\begin{remark}
Every set cover obtained as a solution to the DSCP is a domatic set in graph $\mathcal{R}$.
\end{remark}
\begin{proof}
By definition, a set cover contains all the elements in the universe. Therefore, by the definition of $\mathcal{R}$, the set of vertices $V_i$ that represent a set cover $C_i$ must therefore be connected to all other vertices in the graph $V\backslash V_i$ because $C_i$ shares elements with all other subsets $\mathcal{S}\backslash C_i$.
     \end{proof}
\begin{remark}
Every domatic partition of $\mathcal{R}$ does not correspond to a solution of the DSCP. 
\end{remark}
\begin{proof}
We show this through an example. Consider subsets $\mathcal{S}=\{\{2\},\{1,2\},\{1,2,3\}\}$ and the universe $\mathcal{U}=\{1,2,3\}$. Graph $\mathcal{R}$ therefore has 3 vertices which form a clique, that has three domatic sets. $\mathcal{S}$, however, has only one disjoint set cover. The element $2$ acts as a \emph{proxy} for the other two elements.
     \end{proof}
So we have shown that every DSCP solution is a domatic partition of graph $\mathcal{R}$, but every domatic partition of $\mathcal{R}$ is not a DSCP solution. So, solving the domatic partitioning problem on $\mathcal{R}$ need not necessarily give us a solution to the DSCP. We now propose a polynomial time algorithm for the DSCP - the \emph{GraphPartition} algorithm - that uses domatic partitioning.
\begin{lemma}
The GraphPartition algorithm has an approximation ratio of $2\cdot\ln(\mathcal{|S|}+n)$. \label{Sapprox}
\end{lemma}
\begin{proof}
The GraphPartition algorithm works as follows. By modifying graph $\mathcal{R}$, we make the DSCP and domatic partition problem equivalent, so that existing algorithms for domatic partitioning can be applied to the DSCP.
The first step is to remove the possibility of a \emph{proxy}. We add a collection of \emph{singleton} subsets $\mathcal{T} = \{\:\{1\},\{2\} ,\ldots, \{n\}\:\}$ to $\mathcal{S}$. Let $\mathcal{S'}=\mathcal{T}\bigcup \mathcal{S}$. We now construct graph $\mathcal{R'}$ from $\mathcal{S'}$ and $\mathcal{U}$ as defined in Definition \ref{simplegraph}. We now prove the assertion of Proposition \ref{graphtoset} to advance the proof of Lemma \ref{Sapprox}.
\begin{proposition}
In a domatic partition of $\mathcal{R'}$, all the \hbox{domatic} sets are set covers. \label{graphtoset}
\end{proposition}
\begin{proof}
Each domatic set $d_{i}$ in $\mathcal{R'}$ must have edges connecting it to all the vertices corresponding to the singleton elements $\{ j\}\in \mathcal{T}$ that it does not contain. Hence, the subsets corresponding to the vertices in $d_{i}$ must have all the elements in the universe, and therefore form a set cover. The domatic partition on $\mathcal{R'}$ is therefore the DSCP solved on $\mathcal{S'}$ and $\mathcal{U}$.
     \end{proof}

We see that the minimum degree of $\mathcal{R'}$ is $\delta = F_{min}$, since every singleton set has a degree $F_{min}$ in $\mathcal{R'}$. The number of vertices $m$ in $\mathcal{R'}$ is $|\,\mathcal{S'}|=|\,\mathcal{S}+\mathcal{T}|=|\mathcal{S}|+n$. So, using Lemma \ref{lnndomatic}, the number of disjoint set covers $\{ C_1' , C_2' , C_3' ... \}$ obtained are at least $\frac{F_{min} + 1}{\ln\big(|\mathcal{S}| + n\big)}$.

The set covers $C_i'$ obtained from $\mathcal{R'}$ still have singleton subsets $\{t\}\in \mathcal{T}$ contained in them, which were not part of the original problem. Consider a pairwise combination of two such set covers $C_{k}'' = C_p'\bigcup C_q'$. We see that each element $i\in \mathcal{U}$ is covered at least twice in $C_k''$, since $C_p'$ and $C_q'$ were set covers themselves. Thus, after removing all singletons present in $C_{k}''$, we will be left with another set, say $C_{k}^*$, in which each element is covered at least once. $C_{k}^*$ is therefore a set cover which uses subsets only from $\mathcal{S}$. We can therefore combine all set covers $C_p'$ pairwise and discard the singletons. The number of disjoint set covers obtained is therefore $\frac{F_{min} + 1}{2\ln\big(|\mathcal{S}| + n\big)}$.

An upper bound on the DSCP is $F_{min}$, and so we have an algorithm with approximation ratio $2\cdot\ln(\mathcal{|S|}+n)$, proving Lemma \ref{Sapprox}.
     \end{proof}

Next, we present a hypergraph representation of the problem, which will be useful in the description of later algorithms.
\subsection{Hypergraph Equivalence}
The representation of the universe $\mathcal{U}$ and a set of subsets $\mathcal{S}$ is also possible in hypergraphs. We will use these hypergraphs in two of the algorithms that follow. There are two hypergraph constructions possible - the primal and the dual.
\begin{definition}[Primal Hypergraph $\mathcal{P}(V,E)$]
Each vertex $v\in V$ represents an element $i\in \mathcal{U}$ and each hyperedge $e\in E$ represents a subset $S_j\in\mathcal{S}$. A hyperedge contains a vertex if the corresponding subset $S_j$ contains the element represented by that vertex.
\end{definition}
\begin{definition}[Dual Hypergraph $\mathcal{H}(V,E)$]
Each vertex $v\in V$ represents a subset $S_j\in \mathcal{S}$ and each hyperedge $e\in E$ represents an element $i\in\mathcal{U}$. A hyperedge contains a vertex if the corresponding element was present in the subset $S_j$ represented by that vertex. \label{Dual}
\end{definition}
Refer to Fig. \ref{PrimalDual} for examples of the primal and dual hypergraphs.

Note that a few of the terms defined in Section \ref{terminology} become easy to see in hypergraphs $\mathcal{P}$ and $\mathcal{H}$. $n$ is the number of vertices in $\mathcal{P}$ and the number of hyperedges in $\mathcal{H}$. $|\mathcal{S}|$ is the number of hyperedges in $\mathcal{P}$ and the number of vertices in $\mathcal{H}$. $R$ is the size of the largest hyperedge in $\mathcal{P}$, where the \emph{size} of a hyperedge is the number of vertices it contains. $F_i$ represents the size of hyperedge $i$ in $\mathcal{H}$. $F_{min}$ is represented by the size of the smallest hyperedge in $\mathcal{H}$. $F_{max}$ is the size of the largest hyperedge in $\mathcal{H}$. $\tau_i$ is the hyperedge-degree of hyperedge $i$ in $\mathcal{H}$, where the hyperedge-degree of a hyperedge is defined as the number of other hyperedges with which it shares vertices. $\Delta_{\tau}$ is the maximum hyperedge-degree of $\mathcal{H}$.

In Fig. \ref{PrimalDual}, $\tau_{1}=3$, $\tau_{2}=3$, $\tau_{3}=2$ and $\tau_{4}=2$. $\Delta_{\tau}=3$.\\

\begin{figure}[H]
\centering
  \begin{tabular}{@{}cc@{}}
    \psfrag{1}[][][0.8]{$1$}
    \psfrag{2}[][][0.8]{$2$}
    \psfrag{3}[][][0.8]{$3$}
    \psfrag{4}[][][0.8]{$4$}
    \psfrag{a}[][l][0.9]{(A)}
    \psfrag{\{1,2,4\}}[][l][.6]{$\{1,2,4\}$}
    \psfrag{\{1,3\}}[][r][.6]{$\{1,3\}$}
    \psfrag{\{2,3\}}[][][.6]{$\{2,3\}$}
    \includegraphics[width=.20\textwidth]{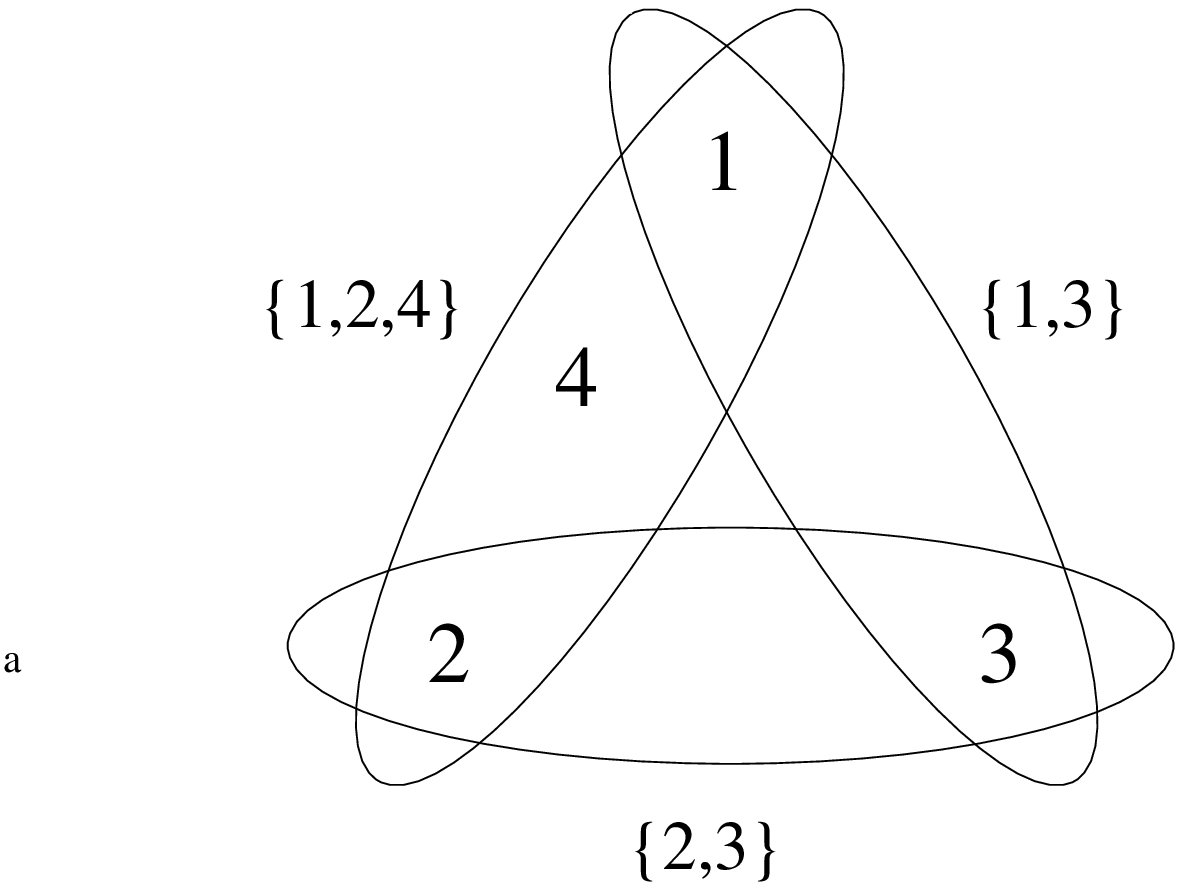}  &
    \hspace{4mm}
    \psfrag{1}[][l][0.8]{$1$}
    \psfrag{2}[][r][0.8]{$2$}
    \psfrag{3}[][][0.8]{$3$}
    \psfrag{4}[][b][0.8]{$4$}
    \psfrag{b}[][b][0.9]{(B)}
    \psfrag{\{1,2,4\}}[][][.48]{$\{1,2,4\}$}
    \psfrag{\{1,3\}}[][][.6]{$\{1,3\}$}
    \psfrag{\{2,3\}}[][][.6]{$\{2,3\}$}
    \includegraphics[width=.20\textwidth]{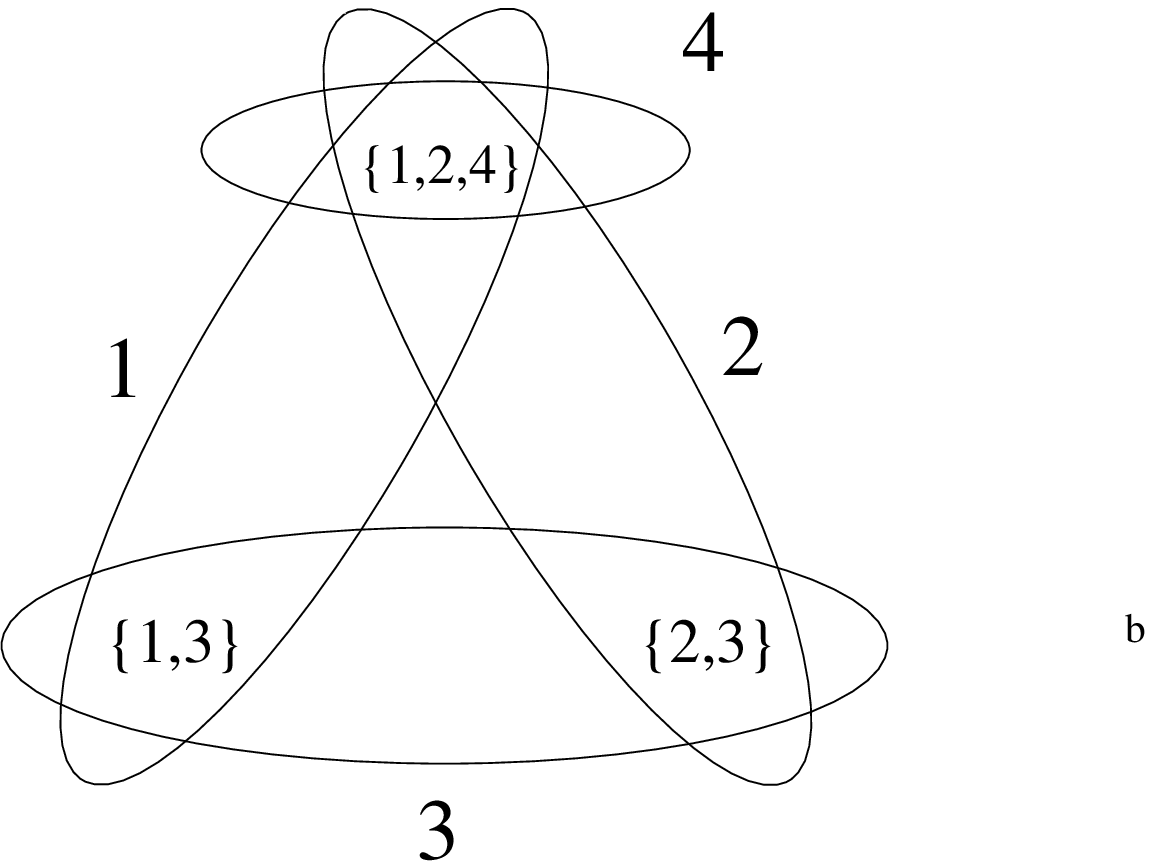} \\
  \end{tabular}
  \caption{Primal (A) and Dual (B) hypergraphs for $\mathcal{U} = \{ 1,2,3,4\} $ and $ \mathcal{S} = \{ \{ 1,2,4\}, \{2,3\}, \{1,3\} \}$} \label{PrimalDual}
\end{figure} 

We now state a definition that we will use to show how the DSCP can be solved on the hypergraph.

\begin{definition}[Polychromatic colouring \cite{bollobas2013cover}]
Poly-chromatic colouring of a hypergraph is defined as the colouring of its vertices with as many colours as possible (say $\ell$) such that each hyperedge contains vertices of all $\ell$ colours.
\end{definition}
\begin{theorem}
DSCP is equivalent to the polychromatic colouring of the dual hypergraph $\mathcal{H}$ (refer Definition \ref{Dual}), in which the subsets corresponding to the vertices of any one colour form a set cover. \label{DSCPpoly}
\end{theorem}
\begin{proof}
Let the dual hypergraph $\mathcal{H}$ be polychromatically coloured with $\ell$ colours. By definition, every colour is present in each hyperedge of $\mathcal{H}$. Let the vertices coloured using a particular colour $a$ be represented by the set $V_a$. Let the corresponding subsets be represented by the set $\mathcal{S}_a$. Since colour $a$ is present in all hyperedges and each hyperedge represents an element $i\in \mathcal{U}$, $\mathcal{S}_a$ must be a set cover. Therefore, every colour corresponds to a set cover. Since each vertex of the hypergraph $\mathcal{H}$, $v\in V$, can only be coloured with one colour, polychromatic colouring ensures that the set covers are disjoint. There are therefore $\ell$ disjoint set covers.
     \end{proof}
It is therefore sufficient to polychromatically colour the hypergraph $\mathcal{H}$ in order to solve the DSCP. We do this by a randomized algorithm which is later derandomized using the method of conditional probabilities.

\subsection{$\ln n$ approximation to the DSCP}
\label{lnnapprox}

Before stating the algorithm itself, we define some notation involved in the analysis.

An \emph{incomplete colouring} of the hypergraph $\mathcal{H}$ is defined as any colouring of all its vertices with $\ell'$ colours in which any hyperedge does not contain all the $\ell'$ colours. Given an incomplete colouring, we define the Boolean variable $A_{e,c}$, which is \emph{true} if edge $e$ does not contain a vertex with the colour $c$ and \emph{false} otherwise. Note that each true $A_{e,c}$ is a \emph{``bad event"}; any colour which is not contained in all edges cannot correspond to a set cover. Conversely, if every hyperedge $e$ contains all $\ell'$ colours, then we have a solution to the DSCP, where the subsets of the same colour represent a set cover. We want to avoid these bad events and achieve a polychromatic colouring with as many colours as possible. We also define $\mathcal{P}[A_{e,c}]$ as the probability that $A_{e,c}$ is true. Let $V(e)$ denote the set of vertices present in hyperedge $e$. Let $[\ell]$ denote the set of colours.  An \emph{invalid} colour is one which is not present in all the hyperedges. Let $L$ be a random variable which denotes the number of invalid colours. Let $\mathbf{E}[X]$ denote the expectation of a random variable $X$.
\begin{theorem}
Any hypergraph $\mathcal{H}$ can be coloured polychromatically with $F_{min}\big(1 - o(1)\big)/\ln n$ colours in polynomial time using a randomized algorithm. \label{Theorem1}
\end{theorem}
\begin{proof}
Let $\ell = F_{min}/\big(\ln(n\ln n)\big)$. Independently colour each vertex with one of $\ell$ colours picked uniformly randomly. The probability of a hyperedge $e$ not containing colour $c$ is given by:
\begin{equation*}
\mathcal{P}[A_{e,c}]=\big(1-1/\ell\big)^{|V(e)|}<e^{-|V(e)|/\ell}\leq e^{-\ln(n\ln n)}=\frac{1}{n\ln n}.
\end{equation*} 
Thus, the expected number of colours which are not present in all hyperedges is
\begin{equation}
\mathbf{E}[L]= \sum_{c\,\in\,[\ell]}\,\sum_{e\,\in \,E}\mathcal{P}[\,A_{e,c}]=\ell\cdot n \cdot \frac{1}{n\ln n} = \frac{\ell}{\ln n}. \label{expectation}
\end{equation}
Thus, the number of colours which form a polychromatic colouring is:
\begin{align}
\ell-\frac{\ell}{\ln n} = \frac{F_{min}}{\ln n}\bigg(1 - \frac{\ln\ln n+1}{\ln(n\ln n)}\bigg). \label{actualapprox}
\end{align} 
The $\ell/\ln n$ vertices coloured with invalid colours do not correspond to set covers. We have therefore obtained \linebreak \hbox{ $\frac{F_{min}}{\ln n}\Big(1 - o(1)\Big)$} disjoint set covers, where the $o(1)$ term $<1/2$ and goes to zero as $n\rightarrow\infty$.
     \end{proof}

The randomized algorithm described by Theorem \ref{Theorem1} can be derandomized using the method of conditional probabilities, so that we always deterministically end up with a polychromatic colouring with $F_{min}\big(1-o(1)\big)/\ln n$ colours.

\textbf{Deterministic Colouring Algorithm :} We again start with $\ell = F_{min}/\big(\ln(n\ln n)\big)$ colours. We first colour all the vertices uniformly randomly with one of the $\ell$ colours. Now, by (\ref{expectation}), $\mathbf{E}[L]=\ell/\ln n$. The algorithm now works by recolouring these vertices deterministically. First, order the vertices arbitrarily as $v_{1},v_{2},v_{3} ,\ldots, v_{|\mathcal{S}|}$ and recolour them one by one in this order. We denote the recolouring of $v_i$ by $c_i$. Let $V^{u}(e)$ denote the set of vertices in hyperedge $e$ which have not yet been recoloured. Now for all hyperedges $e$, the probability that $e$ does not contain colour $c$ given that the vertices $v_{1}, \ldots, v_{i}$ have been recoloured with colours $c_{1}, \ldots, c_{i}$ is $0$ if there exists a vertex in $e$ which has already been recoloured with colour $c$. Otherwise, the probability is given by $\big(1-1/\ell\big)^{|V^{u}(e)|}$, since vertices $v_{i+1},\ldots,v_{|\mathcal{S}|}$ were each coloured uniformly randomly with one of the $\ell$ colours. So,
\begin{equation*} 
	\mathcal{P}[\,A_{e,c}|c_{1}, c_{2}, \ldots, c_{i}]=
	\begin{cases}
	0 , \text{ if } \exists \text{ } q\leq i \text{ such that } v_{q}\in V(e) \\
	\qquad \qquad \qquad \qquad \quad \text{ and } c_{q}=c\\
	\big(1-1/\ell\big)^{|V^{u}(e)|}  \text{ otherwise.}
	\end{cases}
\end{equation*}

Note that, given that the vertices $v_{1}, \ldots, v_{i}$ have been recoloured with colours $c_{1}, \ldots, c_{i}$, the expected number of invalid colours is
\begin{equation}
\mathbf{E}[\,L|c_{1},c_{2} ,\ldots, c_{i}]=\sum_{e\in E} \sum_{c\in [\ell]}\mathcal{P}\big[\,A_{e,c}|c_{1}, c_{2}, \ldots, c_{i}\big].
\end{equation}
We do not recolour $v_1$. Since $v_1$ was coloured uniformly randomly to start with, $\mathbf{E}[L|c_1]=\ell/\ln n$. Now recolour $v_2$ deterministically such that $\mathbf{E}[L|c_1,c_2]\leq \mathbf{E}[L|c_1]=\ell/\ln n$. Similarly, recolour $v_i$ deterministically such that \linebreak $\mathbf{E}[L|c_{1},c_{2} ,\ldots, c_{i-1}, c_{i}]\leq \mathbf{E}[L|c_{1},c_{2} ,\ldots, c_{i-1}]$. It will be shown in the analysis that such a colouring always exists. We will therefore finish the recolouring with $\mathbf{E}[L|c_{1},c_{2} ,\ldots, c_{|\mathcal{S}|}]$ invalid colours, which is a deterministic quantity (since given $c_1,\ldots, c_{|\mathcal{S}|}$, $L$ is a deterministic quantity) that we have ensured through our algorithm to be less than $\ell/\ln n$. We would therefore have deterministically obtained a polychromatic colouring with \hbox{at least $\ell-\ell/\ln n$} colours, as proposed by Theorem \ref{Theorem1}. 

\textbf{Algorithm Analysis :} We now show that it is always possible to deterministically recolour a vertex while ensuring that the expected number of invalid colours does not increase. Note that before the recolouring, every vertex was assigned a colour with probability $1/\ell$. From the linearity of expectation,
\begin{equation*}
\mathbf{E}[L|c_{1},c_{2} ,\ldots, c_{i-1}]= (1/\ell)\sum_{c_{i}\in [\ell]} \mathbf{E}[L|c_{1},c_{2} ,\ldots, c_{i-1}, c_{i}],
\end{equation*}

which is a convex combination. This implies that for some colour $c\in [\ell]$, there exists a colouring of $v_{i}$ with it such that
\begin{equation}
\mathbf{E}[L|c_{1},c_{2} ,\ldots, c_{i-1}, c_{i}]\leq \mathbf{E}[L|c_{1},c_{2} ,\ldots, c_{i-1}].
\end{equation}
Therefore, we always deterministically colour $v_{i}$ such that the expected number of invalid colours \linebreak $\mathbf{E}[L|c_{1},c_{2} ,\ldots, c_{i-1}, c_{i}]\leq \mathbf{E}[L|c_{1},c_{2} ,\ldots, c_{i-1}]$.
The algorithm has a time complexity of $\BigO{\ell^2|\mathcal{S}|nF_{max}}$ in its trivial implementation.

We now propose another polynomial-time algorithm for the DSCP which can be used in certain applications.
\subsection{$\BigO{\ln \Delta_{\tau}}$ Approximation to the DSCP}
\label{lndapprox}
In this section, we present the algorithm to obtain an $\BigO{\ln \Delta_{\tau}}$ approximation to the DSCP, where $\Delta_{\tau}$ is the maximum expansiveness.
A few definitions with respect to hypergraphs will be useful going forward.
\begin{definition}[Path]
In a hypergraph $G(V,E)$, a path between hyperedges $e_i, e_j\in E$ is defined as an alternating sequence of unique hyperedges and vertices $v\in V$, \hbox{$s={e_iv_i\ldots v_je_j}$} such that each hyperedge $e_k\in s$ contains the vertex immediately preceding it and immediately following it in $s$. Note that $e_i$ (the first hyperedge) is only required to contain $v_i$ and $e_j$ (the last hyperedge) is only required to contain $v_j$. Note that this definition resembles the corresponding definition in a graph.
\end{definition}
\begin{definition}[Path length]
In a hypergraph $G(V,E)$, the length of a path between hyperedges $e_i, e_j\in E$ is defined as the number of vertices in the sequence $s$.
\end{definition}
\begin{definition}[Component]
In a hypergraph $G(V,E)$, a component is defined as a set of hyperedges $E_c$ such that for any $e_i, e_j\in E_c$, there exists a path between $e_i$ and $e_j$.
\end{definition}
The $\BigO{\ln \Delta_{\tau}}$ algorithm will also seek to polychromatically colour the hypergraph $\mathcal{H}$ as we did with the $\ln n$ algorithm in Section \ref{lnnapprox}. Again, the subsets corresponding to the vertices of any one colour form a set cover. This time, however, we colour the hypergraph in phases, considering a few vertices to have been coloured correctly at the end of each phase. This corresponds to assigning the corresponding subsets to some set cover. At the end of the algorithm, $F_{min}/(c\ln \Delta_{\tau})$ colours form a polychromatic colouring, and we have $F_{min}/(c\ln \Delta_{\tau})$ disjoint set covers.

We use an algorithmic version of the Lov{\'a}sz Local Lemma introduced by \cite{beck1991algorithmic}, and the analysis is similar to the one in \cite{feige2002approximating}. Without loss of generality, we can consider $\mathcal{H}$ to be connected (otherwise, we can consider each component as a different problem). Our algorithm is divided into $p$ phases (we will define $p$ in Section \ref{pdef}). After each phase, the hypergraph $\mathcal{H}$ is split into components with expected number of hyperedges $\BigO{\Delta_{\tau}^3\ln n}$, where $n$ the number of hyperedges in $\mathcal{H}$ and $\Delta_{\tau}$ is the maximum \emph{expansiveness} of $\mathcal{H}$. At the end of the last phase, we use Theorem \ref{Theorem1} to colour each component in polynomial time.
\subsubsection{Calculating $p$}
\label{pdef}
$p\in \mathbb{N}$ is a number such that $\ln\ln \ldots_{\text{ p times}}n \geq \Delta_{\tau}$ and $\ln\ln \ldots_{\text{ p+1 times}}n < \Delta_{\tau}$. We can see that the function $\ln\ln \ldots_{\text{ p times}}n$ decreases \textbf{very} fast. If the maximum expansiveness drops exponentially, $p$ will only increase by $1$.
\subsubsection{Algorithm Description} We call this algorithm the \hbox{\emph{EXPCover} algorithm} - coverage through expansiveness.

\textbf{Phase 1:} The algorithm works on the dual hypergraph $\mathcal{H}$, which has $\mathcal{|S|}$ vertices and $n$ hyperedges. Number the vertices arbitrarily as $ \{v_1 ,v_2 ,\ldots, v_{\mathcal{|S|}} \} $. Process the vertices in this order, from $v_1$ to $v_{\mathcal{|S|}}$. When a vertex is encountered, colour it uniformly randomly with one of $\ell = F_{min}/\big(c\ln \Delta_{\tau}\big)$ colours. After each colouring, look at the hyperedges to which that vertex belongs and \emph{freeze} a hyperedge (and thus its vertices) if it satisfies the following two properties: \\
\textbf{(i)} The hyperedge has $F_{min}/p$ vertices coloured, where $p$ is as defined in Section \ref{pdef}, \textbf{and}\\
\textbf{(ii)} Not all the $\ell$ colours are present in the hyperedge.\\
If a frozen vertex is encountered during processing, skip it and go on. The first phase will end after all $\mathcal{|S|}$ vertices have been processed.

\textbf{The remaining phases:} After phase 1, there will be three types of hyperedges :\\
\textbf{(i)} \emph{Good} - Contains all the $\ell$ colours.\\
\textbf{(ii)} \emph{Frozen} - Because of the procedure in phase 1.\\
\textbf{(iii)} \emph{Neutral} - Does not contain all $\ell$ colours but was not frozen.

The next phase of the algorithm runs on the component formed by the \emph{Frozen} and \emph{Neutral} hyperedges, which we call \emph{saved}. Since the good hyperedges contain all $\ell$ colours, their colouring can be considered correct, since they will not hinder the polychromatic colouring of the graph. 

We repeat the colouring algorithm defined in phase 1 on the saved components obtained after phase 1, colouring one component at a time. By Theorem \ref{savededges}, the size of the largest saved component at the end of phase 2 will be \linebreak $\BigO{\Delta_{\tau}^3(3\ln \Delta_{\tau} + \ln\ln n)}$.

Thus, after $p$ phases, by the repeated use of Theorem \ref{savededges} $p$ times, we will obtain components of size \linebreak $N = \BigO{\Delta_{\tau}^3(\ln\ln \ldots_{\text{ p times}} n)}$, which is  $\BigO{\Delta_{\tau}^4}$ \linebreak (since $\ln\ln \ldots_{\text{ p times}}n = \BigO{\Delta_{\tau}}$). Note that $\ln N = \BigO{\ln \Delta_{\tau}}$.

Once we have a hypergraph with component size $N$, we apply the algorithm corresponding to Theorem \ref{Theorem1} to colour each component with $\ell=F_{min}/\ln N$ colours and thus achieve an approximation ratio of $\ln N = \BigO{\ln \Delta_{\tau}}$.\\
\subsubsection{Algorithm Analysis}
Our choice of the number of colours $\ell$ and the approximation ratio of the EXPCover algorithm will be made clear by the following Theorem.
\begin{theorem}
It is possible to polychromatically colour a hypergraph $\mathcal{H}$ with $\ell=F_{min}/(c\ln \Delta_{\tau})$ colours, where $c$ is a suitable constant. The EXPCover algorithm therefore provides an $\BigO{\ln \Delta_\tau}$ approximation to the DSCP. \label{lndcolour}
\end{theorem}
\begin{proof}
We will show this by proving that setting $\ell=F_{min}/\big(c\ln \Delta_{\tau}\big)$ in the EXPCover algorithm allows us to choose a particular $c$ such that the size of the largest saved component is bounded. This bound is given by Theorem \ref{savededges}. The choice of constant $c$ itself will be shown in the appendix, in the proof of Theorem \ref{savededges}. 
     \end{proof}
\begin{theorem}
With probability greater than $1/2$, the size of the largest saved component at the end of the first phase is at most $\BigO{\Delta_{\tau}^3\ln n}$. \label{savededges}
\end{theorem}
\begin{proof}
See Appendix \ref{appendix}.
     \end{proof}

We now provide some intuition for the EXPcover algorithm itself. At the end of each phase, an unsaved hyperedge contains all $\ell$ colours. In the last phase of the algorithm with $N$ hyperedges, $\ell-\ell/\ln N$ of the colours will form a polychromatic colouring by Theorem \ref{Theorem1}. Since all $\ell$ colours were present in the good hyperedges which were unsaved in all the phases anyway, the colouring of the last phase also corresponds to a polychromatic colouring of $\mathcal{H}$ with $\ell-\ell/\ln N$ colours. $\ell/\ln N$ invalid colours are only present in the unsaved hyperedges and not in the $N$ saved ones. All that is left to do to show the correctness of the EXPcover algorithm is to prove that $\ell=F_{min}/(c\ln \Delta_{\tau})$ is a valid choice, and that the size of the saved components is as much as we proposed. We do so through the following Theorems.

Also note that the number of saved components at any time can only be as large as $n$, since there are $n$ hyperedges in all. Since the colouring of each component can be achieved in polynomial time, the colouring of all components must therefore be achievable in polynomial time.

We now present a method to deterministically ensure during the algorithm that the size of the saved components is limited.

\textbf{Deterministic Algorithm for Theorem \ref{savededges} :} We need to ensure that we colour the vertices such that the largest saved component is as big as we proposed. We do this through the method of \emph{repeated random colouring}. We colour the vertices uniformly randomly in each phase $p_i$ and then look at the size of the largest saved component after $p_i$ is complete. If it is larger than proposed by Theorem \ref{savededges}, we repeat the phase $p_i$, colouring vertices uniformly randomly again. Note that the probability of appearance of a large component in the first colouring was less than half, and drops exponentially with every repetition of phase $p_i$. So we will end up with a colouring proposed by Theorem \ref{savededges} within a few repetitions of $p_i$.

The time complexity of the EXPCover algorithm is $\BigO{p|\mathcal{S}|n}$, since it runs in $p$ phases and does $n$ operations on each of the $|\mathcal{S}|$ vertices.

The EXPCover algorithm has advantages over the algorithm in Section \ref{lnnapprox} in certain cases. Firstly, its time complexity is better. Secondly, in cases where the network has a low expansiveness, the $\BigO{\ln \Delta_\tau}$ approximation ratio could be less than the $\BigO{\ln n}$ ratio.
  
We now present a third approach to approximating the DSCP, using the Lov{\'a}sz Local Lemma (LLL).

\subsection{Algorithms for the DSCP using the LLL}
We saw from Theorem \ref{Theorem1} that a polychromatic colouring of a hypergraph with $F_{min}\big(1 - o(1)\big)/\ln n$ colours exists. We now provide a new bound using the Lov{\'a}sz Local Lemma (LLL).
\begin{lemma} \label{LLLalgo}
Any hypergraph with $F_{min}\Delta_{\tau}<n$ can be polychromatically coloured with $F_{min}/\ln (\sf{e}$$F_{min}\Delta_{\tau})$ colours, where \sf{e} denotes the base of the natural logarithm.
\end{lemma}
\begin{proof}
See Appendix \ref{LLLproof}.
\end{proof}
It is possible to use Moser and Tardos' results on the algorithmic LLL \cite{moser2010constructive} to achieve a tight algorithmic version of Lemma \ref{LLLalgo}. As shown by Bollobas in \cite{bollobas2013cover}, applying the LLL on the primal hypergraph yields another bound on the number of covers as being $F_{min}/\ln(\sf{e}$$RF_{min}^{2})$, which can also be achieved by the algorithmic LLL \cite{moser2010constructive}.
\subsection{Hardness of the DSCP}
It was shown in \cite{feige2002approximating} that the DSCP is $\ln n$ hard, by reducing the max-3-sat-5-colourability problem to it. A $(1-\epsilon)\ln n$ approximation to the DSCP implies that \linebreak $NP \subseteq DTIME\big(n^{\BigO{\ln \ln n}}\big)$. The colouring algorithm of Section \ref{lnnapprox} therefore achieves the best possible approximation to the MLCP, from Theorem \ref{Theorem1}. By the EXPCover algorithm, we also conclude that the class of problems in which $\Delta_\tau<<n$ is easier than the general MLCP, and can be approximated within a ratio less than $\ln n$.

We now show that the algorithms for DSCP proposed in Sections \ref{lnnapprox} and \ref{lndapprox} are also $\ln n$ and $\BigO{\ln \Delta_{\tau}}$ approximations to MLCP. 
\subsection{$\ln n, \BigO{\ln \Delta_{\tau}}$ Approximations to MLCP}
\begin{lemma}
The $\ln n$ and $\BigO{\ln \Delta_{\tau}}$ approximation algorithms to the DSCP provided in Sections \ref{lnnapprox} and \ref{lndapprox} are also $\ln n$ and $\BigO{\ln \Delta_{\tau}}$ approximation algorithms to the MLCP. \label{MLCPapprox}
\end{lemma}
\begin{proof}
We first note that $F_{min}$ is an upper bound on the MLCP. This is because at least one of the subsets containing the element with frequency $F_{min}$ must be present in all covers, and all those subsets can together be used only for time $F_{min}$. Thus, the algorithm in Section \ref{lnnapprox}, which returns at least $F_{min}/\ln n$ disjoint set covers can be used by turning on each set cover for time $1$, and achieve a network lifetime of at least $F_{min}/\ln n$. Thus we get a $\ln n$ approximation to MLCP. The algorithm in Section \ref{lndapprox} returns at least $F_{min}/\BigO{\ln \Delta_{\tau}}$ disjoint set covers, and is effectively an $\BigO{\ln \Delta_{\tau}}$ approximation to the MLCP.
     \end{proof}
\subsection{Advantages of our MLCP algorithms} \label{advantage}
The algorithm for solving the MLCP in \cite{berman2004power} returns a network lifetime of \hbox{OPT/$(1+\ln n)$}, where OPT represents the optimal solution of the MLCP. The algorithm in Section \ref{lnnapprox}, on the other hand, returns a network lifetime of $F_{min}/\ln n$ for large $n$. Since \hbox{$F_{min}\geq$ OPT}, our algorithm performs better than the algorithm in \cite{berman2004power} for large $n$.

In addition, using the DSCP to solve the MLCP has the following advantages. Firstly, with the DSCP approach, any sensor need not be repeatedly turned on and off, in contrast to the non-disjoint approach of \cite{berman2004power}. Repeated turn-on and offs could drain the sensor battery and using the DSCP approach circumvents that problem.

Secondly, the DSCP solution can be trivially extended to the $k$-coverage problem \cite{gao2006sensor}. The $k$-coverage problem is to ensure that the lifetime of the network is maximized while covering all targets by at least $k$ active sensors at all times. It is not difficult to see that an upper bound on the $k$-coverage problem is $F_{min}/k$, by the same logic that was used for the general MLCP. Our DSCP solution to the MLCP problem in Section \ref{lnnapprox} returns $F_{min}/\ln n$ set covers. We partition these set covers into groups of $k$ set covers each, and call each group a $k$-cover. Note that the number of $k$-covers provides a solution to the $k$-coverage problem. The number of $k$-covers obtained is $F_{min}/k\ln n$, and the approximation ratio for the $k$-coverage problem is therefore $\ln n$.
We can therefore see that the DSCP solution affords an easy extension to the $k$-coverage problem, unlike non-disjoint set cover solutions. Another logical extension is that of variable $k$-coverage, in which each target needs to be covered by $k$ sensors at a particular time, but this $k$ can now vary over time. A typical example of this is in sensing forest fires, or oil pipelines, in which a larger $k$ would be needed in the dry, hot seasons, and a lower $k$ in cooler, wetter seasons. Note that our DSCP algorithm extends to this case of variable $k$-coverage as well, with the same approximation ratio of $\ln n$.

We now look at the MLCP in special, geometric settings and present algorithms for these cases.
\subsection{Algorithms for the MLCP in the Geometric Case}
In 1 dimension (1-D), the \emph{geometric case} refers to problems in which each sensor covers contiguous targets on a line. In 2 dimensions (2-D), each sensor is assumed to cover a circular region around itself. We first look at 1-D, for which we derive a polynomial time optimal algorithm for the MLCP. We then revisit a result in \cite{ding2012constant} for 2-D and improve upon its approximation algorithm.
\subsubsection{1-D}
In this case, the targets are assumed to lie on a line, and each sensor can only monitor a contiguous interval on that line. We again denote the $n$ targets by a universe of elements $\mathcal{U}=\{1,2,\ldots,n\}$, ordered from left to right on the line. We denote the sensors by a set of subsets $\mathcal{S}=\{S_1,S_2\ldots\}$. We also denote the left-most and right-most elements covered by subset $S_j$ by $\sf{L}$$(S_j)$ and $\sf{R}$$(S_j)$, respectively. We say that a subset $S_j$ \emph{covers all elements} from $\sf{L}$$(S_j)$ to $\sf{R}$$(S_j)$. Let the subsets containing a particular element $i\in \mathcal{U}$ be represented by $\mathbf{D}_i=\{S_j:i\in S_j\}$. We first define a \emph{directed neighbour} of a subset.
\begin{definition}[Directed Neighbour]
\hspace{-1mm}A subset $S_j$ is said to be a directed neighbour of subset $S_k$ if \textbf{(i)} $\sf{L}$$(S_k)\leq \sf{L}$$(S_j)$ and \textbf{(ii)} $S_k\bigcup S_j$ covers all elements from $\sf{L}$$(S_k)$ to $\sf{R}$$(S_j)$. \label{directedneighbour}
\end{definition}
In other words, a subset's directed neighbour lies on its right and they together cover a set of contiguous targets. 
We formulate the MLCP as a maximum flow problem on a network $N(V,E)$ defined as follows:

\textbf{Construction of $N(V,E)$ :} Let each vertex $v_j\in V$ represent a subset $S_j\in \mathcal{S}$. Add two more vertices $s,t$ to $V$, to represent the source and sink respectively. Connect two vertices $v_a,v_b\in V\backslash\{s,t\}$ with a directed edge from $v_a$ to $v_b$ if $S_b$ is a \emph{directed neighbour} of $S_a$. Each vertex $v\in V\backslash\{s,t\}$ is defined to have a capacity $1$, as a result of the corresponding sensor having a battery capacity $1$. Now, connect $s$ to the vertices corresponding to all subsets $S_j\in \mathbf{D}_1$ with directed edges originating at $s$. Similarly, connect the vertices corresponding to all subsets $S_j\in \mathbf{D}_n$ to $t$ with directed edges terminating at $t$.  
\begin{figure}[H]
    \begin{center}
    \psfrag{1}[][][0.9]{\footnotesize{1}}
    \psfrag{2}[][][0.9]{\footnotesize{2}}
    \psfrag{3}[][][0.9]{\footnotesize{3}}
    \psfrag{4}[][][0.9]{\footnotesize{4}}
    \psfrag{5}[][][0.9]{\footnotesize{5}}
    \psfrag{6}[][][0.9]{\footnotesize{6}}
    \psfrag{7}[][][0.9]{\footnotesize{7}}
    \psfrag{8}[][][0.9]{\footnotesize{8}}
    \psfrag{9}[][][0.9]{\footnotesize{8}}
    \psfrag{10}[][][0.9]{\footnotesize{9}}
    \psfrag{11}[][][0.9]{\footnotesize{10}}
    \psfrag{12}[][][0.9]{\footnotesize{11}}
    \psfrag{13}[][][0.9]{\footnotesize{12}}
    \psfrag{S1}[][][0.9]{$S_1$}
    \psfrag{S2}[][][0.9]{$S_2$}
    \psfrag{S3}[][][0.9]{$S_3$}
    \psfrag{L}[][l][0.9]{$\sf L$$(S_1)$}
    \psfrag{R}[][l][0.9]{$\sf R$$(S_2)$}
    \includegraphics[scale=.4]{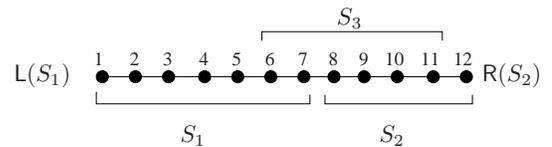}
    \caption{A 1-D example. $S_2$ and $S_3$ are directed neighbours of $S_1$. $S_2$ is a directed neighbour of $S_3$. There will be directed edges $v_1\rightarrow v_2, v_1\rightarrow v_3, v_3\rightarrow v_2$ in the corresponding flow network.} \label{1Dfig}
    \end{center}
\end{figure}
\begin{lemma}
Every path from $s$ to $t$ in $N$ is a set cover of $\mathcal{U}$ from $\mathcal{S}$, and vice versa. \label{pathsetcover}
\end{lemma}
\begin{proof}
To prove Lemma \ref{pathsetcover}, we first prove the following claim: For a path $p_i\in N$ denoted by the ordered set of vertices in $V$ represented by $V_i=\{v_{i_1},v_{i_2},\ldots, v_{i_k}\}$, the set of subsets $Q_i=\{S_{i_1}, S_{i_2},\ldots, S_{i_k}\}$ covers all elements from $\sf{L}$$(S_{i_1})$ to $\sf{R}$$(S_{i_k})$. We will prove this claim by induction. 

\emph{Base Case :} For a path defined by $V_i=\{v_{i_1},v_{i_2}\}$, $S_{i_2}$ must be a directed neighbour of $S_{i_1}$. Therefore, by Definition \ref{directedneighbour}, $Q_i=\{S_{i_1}, S_{i_2}\}$ must cover all elements from $\sf{L}$$(S_{i_1})$ to $\sf{R}$$(S_{i_2})$. So the claim is true for the base case.
\emph{Assume} the claim is true for a path defined by $V_i^k=\{v_{i_1},v_{i_2},\ldots v_{i_k}\}$, i.e. $Q_i^k=\{S_{i_1}, S_{i_2},\ldots, S_{i_k}\}$ covers all elements from $\sf{L}$$(S_{i_1})$ to $\sf{R}$$(S_{i_k})$. 
\emph{Induction Step :} Now, for \linebreak $V_i^{k+1}=\{v_{i_1},\ldots v_{i_k}, v_{i_{k+1}}\}$, $S_{i_{k+1}}$ must be a directed neighbour of $S_{i_k}$, and so $Q_i^*=\{S_{i_k},S_{i_{k+1}}\}$ must cover all elements from $\sf{L}$$(S_{i_k})$ to $\sf{R}$$(S_{i_{k+1}})$. Since $\sf{L}$$(S_{i_k})\leq \sf{R}$$(S_{i_k})$, and elements from $\sf{L}$$(S_{i_1})$ to $\sf{R}$$(S_{i_k})$ have already been covered by $Q_i^*$, $Q_i^{k+1}=\{S_{i_1},\ldots S_{i_{k+1}}\}$ must cover all elements from $\sf{L}$$(S_{i_1})$ to $\sf{R}$$(S_{i_{k+1}})$. The induction is therefore complete.

Let us consider the set of paths $P$ such that any path $p\in P$ is from some vertex $v_a$ such that $S_a\in \mathbf{D}_1$ to some vertex $v_b$ such that $S_b\in \mathbf{D}_n$. Since the above induction has been proven to be true, vertices in $p$ must together cover all elements from $\sf{L}$$(S_a)=1$ to $\sf{R}$$(S_b)=n$. Therefore, $p$ forms a set cover on $\mathcal{U}$ using $\mathcal{S}$. Now, any path from $s$ to $t$ must include a path $p\in P$. Therefore, every path from $s$ to $t$ forms a set cover.

The proof for the reverse direction is also similar, and has been omitted for the sake of brevity.
     \end{proof}

\begin{lemma}
The MLCP on $\mathcal{U}$ and $\mathcal{S}$ is equivalent to the maximum flow through $N$ from $s$ to $t$. \label{MLCPmaxflow}
\end{lemma}
\begin{proof}
The maximum flow through $N$ can be split up into flows such that flow $f_i$ exists in path $p_i$ from $s$ to $t$, each of which is a set cover on $\mathcal{U}$ using $\mathcal{S}$ by Lemma \ref{pathsetcover}. Therefore, each flow $f_i$ corresponds to using set cover $C_i$ for time $t_i=f_i$. Since the maximum flow problem can be formulated as $\max\sum_if_i$ subject to the vertex capacity constraints, this is equivalent to $\max\sum_it_i$ subject to the constraints on battery capacity, which is the solution to the MLCP.
     \end{proof}

The maximum flow problem is a classical network problem which is known to have an optimal polynomial time algorithm \cite{west2001introduction}. The MLCP solution can therefore be arrived at in polynomial time.
We now state a Theorem by West from \cite{west2001introduction}, to claim that the solution to the MLCP is equal to the solution of the DSCP.
\begin{theorem}[Integrality Theorem \cite{west2001introduction}]
If all edge or vertex capacities in a network are integers, there is a maximum flow assigning integer flow to each edge. \label{westtheorem}
\end{theorem}

It is also shown in \cite{west2001introduction} that one can algorithmically arrive at an integral solution to the maximum flow problem in polynomial time.
\begin{lemma}
In 1-D, the MLCP and DSCP solutions are equal. \label{1Dproof}
\end{lemma}
\begin{proof} 
By Lemma \ref{MLCPmaxflow}, we know that the MLCP can be solved as a maximum flow problem on $N$, which has vertex capacities equal to $1$.  By Theorem \ref{westtheorem}, an optimal solution to the maximum flow problem on $N$ involves integer flows. But an integral flow on $N$ ensures that every vertex is present on at most one path from $s$ to $t$. This is because the vertex capacity constraints force the condition that for every vertex, a maximum of one incoming edge and one outgoing edge can have flow $1$. This is effectively a solution to the DSCP, since each path represents a set cover and each vertex a subset. So the DSCP solution is equal to the MLCP solution.
     \end{proof}

In \cite{berman2004power}, it was conjectured that if sensor coverage areas are convex, the solution of the MLCP cannot exceed the solution of the DSCP by more than 50\%. Lemma \ref{1Dproof} tightens this conjecture for the case of 1-dimensional coverage. Another consequence of Lemmas \ref{1Dproof} and \ref{MLCPmaxflow} is the following Theorem.
\begin{theorem}
In 1-D, the DSCP can be optimally solved in polynomial time.
\end{theorem}
By using the algorithm which implements Theorem \ref{westtheorem} on the network $N$, the DSCP can be solved in polynomial time.
\subsubsection{Circular Sensor Coverage - 2-D}
\begin{theorem}
There exists a $1+\epsilon$ approximation to the MLCP when sensor coverage areas are circular.
\end{theorem}
\begin{proof}
It was shown in \cite{berman2004power} that the solution of the MLCP through the Garg-Koenemann algorithm has an approximation ratio equal to $(1+\epsilon)X$ , where $X$ represents the approximation ratio of finding the minimum weight set cover. Without geometric restrictions, $X=\ln n$, which was shown in \cite{feige1998threshold} to be the best possible. In the case where sensor coverage areas are circular, however, \cite{mustafa2010improved} shows that $X=(1+\epsilon)$. It therefore follows that by following the steps illustrated in \cite{berman2004power} using the Garg-Koenemann algorithm, one can obtain an approximation ratio of $(1+\epsilon)$ by using the algorithm in \cite{mustafa2010improved} to find the set covers.
     \end{proof}
This is an improvement on the result in \cite{ding2012constant}, which provides a $(4+\epsilon)$ approximation to the MLCP in 2-D.
\begin{figure}[h!]
    \begin{center}
    \psfrag{0}{\footnotesize{$0$}}
    \psfrag{3}{\footnotesize{$3$}}
    \psfrag{6}{\footnotesize{$6$}}
    \psfrag{9}{\footnotesize{$9$}}
    \psfrag{12}{\footnotesize{$12$}}
    \psfrag{15}{\footnotesize{$15$}}
    \psfrag{40}{\footnotesize{$40$}}
    \psfrag{65}{\footnotesize{$65$}}
    \psfrag{90}{\footnotesize{$90$}}
    \psfrag{115}{\footnotesize{$115$}}
    \psfrag{140}{\footnotesize{$140$}}
    \psfrag{165}{\footnotesize{$165$}}
    \psfrag{190}{\footnotesize{$190$}}
    \psfrag{aaaaaaaa}[][][.65]{$\ln n(1+o(1))$}
    \psfrag{eeeeeeee}[][][.65]{$\ln n$\qquad\qquad}
    \psfrag{cccccccc}[][][.7]{Simulation\:\:}
    \psfrag{xlabel}[b][][0.9]{Number of targets $n$}
    \psfrag{ylabel}[][b][0.9]{Approximation Ratio $\rho$}
    \includegraphics[scale=.58]{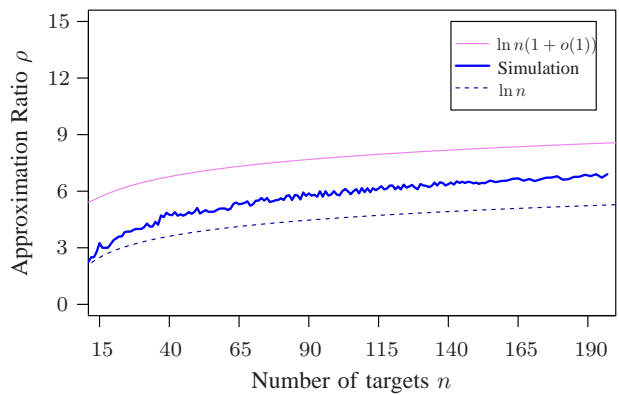}
    \caption{Simulated approximation ratio $\rho$ of the derandomized algorithm for MLCP in Section \ref{lnnapprox}} \label{simulation}
    \end{center}
\end{figure}
\section{Simulations}
We simulated the derandomized algorithm to approximate the MLCP discussed in Section \ref{lnnapprox} on universes $\mathcal{U}$ of different sizes with a large number of subsets of random size picked uniformly randomly from the elements of $\mathcal{U}$. Let the network lifetime yielded by each simulation be $z$. Over multiple simulations, the approximation ratio relative to $F_{min}$, i.e. $\rho=F_{min}/z$ is plotted against $n$ in Fig. \ref{simulation}. It is not possible to plot the approximation ratio relative to the MLCP's optimal solution, since the brute force algorithm takes exponentially longer as $n$ increases. Also plotted in Fig. \ref{simulation} are the functions $\ln n(1+o(1)) = \frac{\ln n}{1-\frac{\ln\ln n +1}{\ln (n\ln n)}}$, which is the actual approximation ratio of the algorithm by (\ref{actualapprox}), and $\ln n$, which is theoretically the minimum possible, by our results in Section \ref{MLCPishard}. Note that our algorithm never does as well as the theoretical bound, since $\rho\rightarrow\ln n$ only as $n \rightarrow \infty$. However, $\rho$ is always bounded between $\ln n$ and $\ln n(1+o(1))$, since $\max \rho=\frac{\ln n}{1-\frac{\ln\ln n +1}{\ln (n\ln n)}}$. The simulated results are therefore concurrent with the theoretical results and reinforce our claim that our algorithm for the MLCP is the best possible polynomial-time algorithm within constant factors.

\section{Conclusions}
We showed that the MLCP is $\ln n$ hard to approximate, thereby showing the optimality of all $\ln n$ approximation algorithms. We then derived a new $\ln n$ approximation algorithm to the MLCP using an approximation to the DSCP. We showed that our algorithm has several advantages over existing approximation algorithms. The simulations of our algorithm support our hardness results. We also showed that the MLCP in the 1-dimensional case could be solved in polynomial time, and that that solution is equal to the DSCP solution. We also improved upon a result in \cite{ding2012constant} to show that a $1+\epsilon$ approximation algorithm for the MLCP exists for any $\epsilon>0$ for the case where sensor coverage areas are circular.
\bibliographystyle{IEEEtran}
\bibliography{research}
\begin{appendices}
\section{DSCP algorithm when $F_{max}=F_{min}=2$} \label{fmaxfminappendix}
If $F_{max}=F_{min}=2$, then $F_i=2$, $\forall$ $i\in \mathcal{U}$. The solution to the DSCP can therefore be either 1 or 2. Note that our algorithm will operate on graph $\mathcal{R}(V,E)$ as defined in Definition \ref{simplegraph}.

We state our result through the following Theorem.
\begin{theorem}
If $F_{max}=F_{min}=2$, the optimal solution to the DSCP can be found in polynomial time through the 2-colouring of $\mathcal{R}(V,E)$.
\end{theorem}
\begin{proof}
Note that the 2-colouring of any graph can be accomplished in polynomial time. To prove the theorem, we prove that:
\begin{lemma}
Given $F_{max}=F_{min}=2$, graph $\mathcal{R}(V,E)$ is 2-colourable if and only if 2 disjoint set covers exist in $\mathcal{S}$. \label{2colour}
\end{lemma}
\begin{proof}
Let the colours used be red(R) and blue(B). The two colours are said to be \emph{opposite} to each other. We denote the red vertices by set $V_R$ and the blue vertices by set $V_B$. The corresponding sets of subsets are denoted by $\mathcal{S}_R$ and $\mathcal{S}_B$. 

Consider any element $i\in \mathcal{U}$. Recall that $F_i=2$, $\forall$ $i\in \mathcal{U}$. Let $i$ be present in two subsets $S_p$ and $S_q$. The vertices $v_p$ and $v_q$ are therefore connected by an edge. In a valid 2 colouring, $v_p$ and $v_q$ must be coloured using opposite colours. Let us assume without loss of generality that $v_p$ was coloured red and $v_q$ blue. As a consequence, $S_p\in \mathcal{S}_R$ and $S_q\in \mathcal{S}_B$. Now for element $i\in \mathcal{U}$, both $\mathcal{S}_R$ and $\mathcal{S}_B$ must contain subsets  with every $i\in \mathcal{U}$. Both $\mathcal{S}_R$ and $\mathcal{S}_B$ are therefore set covers. Therefore, if $\mathcal{R}$ is 2-colourable graph, there exist 2 disjoint set covers given by $\mathcal{S}_R$ and $\mathcal{S}_B$.

We now show that if $\mathcal{R}$ is not 2-colourable, then only 1 disjoint set cover exists. By virtue of not being 2-colourable, any colouring of $\mathcal{R}$ with red and blue colours must result in some two neighbouring vertices $v_p$ and $v_q$ having the same colour, say red. Now $S_p$ and $S_q$ must share at least one element, say $a$. But both $S_p,S_q\in \mathcal{S}_R$. Since $F_a=2$, no subset in $\mathcal{S}_B$ can therefore contain element $a$. So only one set cover exists, represented by $\mathcal{S}_R$.  
     \end{proof}
In the case where $F_{max}=F_{min}=2$, a successful \hbox{2-colouring} of graph $\mathcal{R}$ returns 2 set covers. If a \hbox{2-colouring} is impossible, only 1 disjoint set cover exists. So we use a 2-colouring algorithm on graph $\mathcal{R}$. 
     \end{proof}

\section{Proof of Lemma \ref{LLLalgo}} \label{LLLproof}
To prove Lemma \ref{LLLalgo}, we first review the LLL itself.
\begin{lemma}[LLL \cite{erdos1975problems}]
Suppose we are given $k$ events $T_{1},T_{2} ,\ldots, T_{k}$ such that the probability of occurrence of any event $\mathcal{P}[T_{i}]<p$, $\forall$ $i$. Let the occurrence of any event be dependent on the occurrence of at most $d$ other events. If $\sf{e}\,$$\times\, p\times(d+1)\leq 1$, then the probability that all those events do not occur $\mathcal{P}(\bigcap \bar{T_{i}})>0$. Essentially, given that condition, there is a non-zero probability that none of the events occur. 
\end{lemma}

Independently colour each vertex of $\mathcal{H}$ with one of $\ell=F_{min}/\ln (\sf{e}$$ F_{min}\Delta_{\tau})$ colours. Notice that for edge $e$ and colour $c$,
\begin{equation}
\small
\mathcal{P}[A_{e,c}] \leq (1-1/\ell)^{|V(e)|} \leq (1-1/\ell)^{F_{min}} \leq 1/(\Delta_{\tau}F_{min}\sf{e})
\end{equation}
\normalsize
Now each event $A_{e,c}$ is independent of all but $\ell\times \Delta_{\tau}$ other such events if all edges and all colours are considered, since each hyperedge has a maximum of $\Delta_{\tau}$ other hyperedges as its neighbours.
Now $\ell<F_{min}$, so we can see that the degree of dependence $d$ in the LLL is less than $F_{min}\Delta_{\tau}$. So we set $d=F_{min}\Delta_{\tau}-1$ and $p=1/\sf{e}$$F_{min}\Delta_{\tau}$ in the LLL to prove Lemma \ref{LLLalgo}.

\section{Size of Saved Components}
\label{appendix}
\textbf{Proof of Theorem \ref{savededges} :} In this section, we look at the size of the saved components at the end of each phase of the EXPCover algorithm. Here, \emph{size} of a component refers to the number of hyperedges present in it. To refresh, we first redefine a saved component:
\begin{definition}
A saved component contains Frozen and Neutral hyperedges.
\end{definition}
After each phase of the algorithm, the hyperedges are categorized into 3 sections: Good (correctly coloured), Frozen and Neutral.

We shall calculate the size of the largest saved component formed due to \emph{Frozen} and \emph{Neutral} edges in terms of $n$ and $\Delta_{\tau}$. We divide the proof into two parts:\\

\textbf{Part 1: Frozen vertices :} Let $X(e_{i})$ be an indicator random variable for a hyperedge $e_{i}$ to be frozen. Define \hbox{$q=\ell(1 - 1/\ell)^{F_{min}/p}.$}
\begin{lemma} 
Consider $E = \{e_{1}, e_{2}, \ldots, e_{k}\}$ such that any two edges $e_{i} , e_{j} \in E$, have no common vertices. Then the probability that all hyperedges $e\in E$ are frozen \linebreak \hbox{$P \big( X(e_{1}) = X(e_{2}) = X(e_{3}) =\ldots = X(e_{k}) = 1\big)$} is $q^k$.
\end{lemma}
\begin{proof}
A hyperedge $e_{i}$ will be frozen, by definition, when $F_{min}/p$ of its vertices are coloured and the hyperedge lacks at least one of the $\ell$ colours. The probability of such an event is 
\begin{equation}
P(X(e_{i})) = \ell(1 - 1/\ell)^{F_{min}/p} = q. \label{frozen1}
\end{equation}
By definition of the set $E$, $e_{i},e_{j}\in E$ have no common vertices. Hence, $X(e_{i})$ and $X(e_{j})$ are independent. Thus \hbox{$P (X(e_{i}) = X(e_{j}) = 1) = P(X(e_{i} = 1)) \cdot P(X(e_{j} = 1))$}. Since this independence holds for all the edges in $E$, we obtain the following.
\begin{equation*}
\begin{split}
P \big(X(e_{1})=\ldots=X(e_{k})= 1\big)&=\prod_{i\in \{1,2,\ldots,k\}}P\big(X(e_{i})=1\big)\\
					&= q^k.
\end{split}
\end{equation*}
     \end{proof}

\textbf{Part 2 : Saved edges :}
\begin{definition}
Distance between two hyperedges $e_{i}$ and $e_{j}$ is defined to be the length of the shortest path between them.
\end{definition}
\begin{definition}
A set of hyperedges is said to be a \emph{3-separated} set if the distance between any two hyperedges in the set is at least \emph{3}.
\end{definition}

We denote by $Y(e_i)$, an indicator random variable, the \emph{status} of hyperedge $e_i$. $Y(e_i)=1$ if $e_i$ is saved, and $Y(e_i)=0$ otherwise.
\begin{lemma}
Let $T =  \{e_{1}, e_{2}, \ldots, e_{k}\}$ be a 3-separated set. Then the probability that all hyperedges $e\in T$ are saved $P \big(Y(e_{1}) = Y(e_{2}) = \ldots = Y(e_{k}) = 1\big)$ is at most $(2\Delta_{\tau}q)^k$. \label{lemma3}
\end{lemma}
\begin{proof}
Any hyperedge $e$ in the hypergraph $\mathcal{H}$ will be saved if \textbf{(i)} it is frozen or \textbf{(ii)} it is rendered neutral because its vertices are frozen by neighbouring hyperedges. Note that a non-frozen edge will not be a saved edge if none of its neighbours are frozen. Let $N(e)$ be the set of edges which share vertices with $e$, i.e. the neighbours of $e$. We see that
\begin{equation}
P\big(Y(e) = 1 \big) \leq P\big(X(e) = 1 \big) + \sideset{}{} \sum_{e' \in N(e)} P\big(X(e') = 1 \big). \label{saved1}
\end{equation}
Using the fact that $\max|N(e)|=\Delta_{\tau}$ and equation \eqref{frozen1}, we can bound \eqref{saved1}
\begin{equation}
P\big(Y(e) = 1 \big) \leq \: q + \Delta_{\tau}q \: < 2\Delta_{\tau}q. \label{saved2}
\end{equation}
The hyperedges in $T$ do not have common neighbours because they are \emph{3-separated}. Thus, $\forall$ $e_i,e_j\in T$, we observe that $Y\big(e_{i}\big)$ and $Y\big(e_{j}\big)$ are independent. Using equation \eqref{saved2} and the independence of $Y\big(e_i)$ and $Y(e_j)$, $\forall$ $i\neq j$, we conclude that
\begin{equation*}
\begin{split}
P \big(Y(e_{1})=\ldots=Y(e_{k})=1\big)&=\prod_{i\in \{1,2,\ldots,k\}}P\big(Y(e_{i}=1)\big)\\
	       				&\leq (2\Delta_{\tau}q)^k.
\end{split}
\end{equation*}

     \end{proof}

So we have shown that the probability that a 3-separated set of size $k$ is saved is less than $(2\Delta_{\tau}q)^{k}$. We will now look at the number of 3-separated sets of size $k$ which can exist in our graph.
\begin{lemma} 
For the hypergraph $\mathcal{H}$ with $n$ hyperedges, the number of \emph{3-separated} sets of size $k$ is at most $n(4\Delta^3)^k$. \label{numberof3sep}
\end{lemma}
\begin{proof}
The number of distinct shapes of a tree with $k$ hyperedges is at most $4^{k-1}$. This is obtained by ordering the hyperedges and assigning two flag bits to each hyperedge \textbf{(i)} if it has same parent as the previous one (or not) and \textbf{(ii)} if it has a child (or not). This gives rise to $4^{k-1}$ possibilities.\\
For the hypergraph $\mathcal{H}$, there are $n$ ways of choosing the root and for each successive hyperedge there are at most $\Delta_{\tau}$ choices. For a 3-separated set, this leaves a maximum of $\Delta_{\tau}^3$ choices. Thus, the number of distinct \emph{3-separated} sets of size $k$ is at most $n(4\Delta_{\tau}^3)^k$.

     \end{proof}

Combining Lemma \ref{lemma3} and Lemma \ref{numberof3sep}, we see that the probability that any 3-separated set of size $k$ exists is less than $(2\Delta_{\tau}q)^{k} \times n(4\Delta_{\tau}^3)^k=n(8\Delta_{\tau}^4q)^k$.\\
Recall that $q = \ell(1-1/\ell)^{k}$ and $\ell = F_{min}/c\ln\Delta_{\tau}$. We choose $c$ such that $q = \Delta_{\tau}^{-5}/8$ and choose $k = \frac{\log(2n)}{\log(\Delta_{\tau})}$ and thereby ensure that
\begin{equation}
n(8\Delta_{\tau}^4q)^k < 1/2. \label{final3}
\end{equation}

We complete the proof by showing that connected hypergraphs of a particular size must have a 3-separated set of size k. From (\ref{final3}), the probability that the size of a saved component is at most this size is greater than $1/2$.
\begin{lemma}
Any connected hypergraph with $k\Delta_{\tau}^3$ hyperedges \textbf{must} have a \emph{3-separated set of size $k$}.
\end{lemma}
\begin{proof}
Each hyperedge can have a maximum of $\Delta_{\tau}$ hyperedge neighbours. It can therefore only have $\Delta_{\tau}^3$ hyperedges at a distance $3$ from it. The proof therefore follows. 
     \end{proof}

Note : We can arbitrarily increase the probability of limiting the size of components formed by the \emph{saved} vertices to be at most $k\Delta_{\tau}^3$, where $ k = \BigO{\ln n}$

\end{appendices}

\end{document}